\PassOptionsToPackage{dvipsnames}{xcolor}

\newcommand\MFCSSHORTER[1]{}

\documentclass[apxproof]{lipics-v2021}

\providecommand\apxparameter{appendix=inline}
\usepackage[\apxparameter]{apxproof}

\usepackage{mathptmx}
\usepackage[utf8]{inputenc}
\usepackage[T1]{fontenc}

\usepackage{version}
\excludeversion{SIDE}
\excludeversion{SIDEE}
\excludeversion{PASFINALISE}
\excludeversion{LONGUER}
\excludeversion{AVOIR}

\nolinenumbers

\newcommand\finprovisoire{\bibliographystyle{plainurl}
\bibliography{/Users/bournez/bibliographie/BIBDESK-DIR=/bournez,/Users/bournez/bibliographie/BIBDESK-DIR=/perso}
\end{document}}

\newcommand\lipicsmode[1]{#1}
\newcommand\apxproofmode[1]{#1}
\newcommand\apxproofmodesubsection[1]{#1}

\providecommand\danger{}

\providecommand\apxproofmode[1]{}
	
\providecommand\apxproofmodesubsection[1]{}

\providecommand\lipicsmode{}
		
\newcommand\nolipics[1]{#1}
\lipicsmode{\renewcommand\nolipics[1]{}}

\providecommand\NODECORATION[1]{}
\lipicsmode{
\providecommand\NODECORATION[1]{#1}
}

\providecommand\PASSIAPXPROOF[1]{#1}
\apxproofmode{\renewcommand\PASSIAPXPROOF[1]{}}

\providecommand\ifnotTDEXAM[1]{#1}
\providecommand\ifnotSLIDE[1]{#1}

\providecommand\PREUVESENAPPENDIX[1]{}

\NODECORATION{
\usepackage{fourier}
}

\usepackage[framemethod=tikz]{mdframed}

\newcommand\ENVCOLOR[1]{#1}
\newcommand\BEGINTEMPENVCOLOR[1]{\let\envcolorwas\ENVCOLOR\renewcommand\ENVCOLOR[1]{#1}}
\newcommand\ENDTEMPENVCOLOR{\let\ENVCOLOR\envcolorwas}

\NODECORATION{
\surroundwithmdframed[hidealllines=true,backgroundcolor=\ENVCOLOR{blue!10},roundcorner=8pt]{exercice}
\surroundwithmdframed[hidealllines=true,backgroundcolor=\ENVCOLOR{blue!10},roundcorner=8pt]{exercise}
\surroundwithmdframed[hidealllines=true,backgroundcolor=\ENVCOLOR{blue!10},roundcorner=8pt]{exercicee}
\surroundwithmdframed[hidealllines=true,backgroundcolor=\ENVCOLOR{blue!10},roundcorner=8pt]{exerciced}
\surroundwithmdframed[hidealllines=true,backgroundcolor=\ENVCOLOR{blue!10},roundcorner=8pt]{exercicec}
\surroundwithmdframed[hidealllines=true,backgroundcolor=\ENVCOLOR{blue!10},roundcorner=8pt]{exercicedc}
\surroundwithmdframed[hidealllines=true,backgroundcolor=\ENVCOLOR{orange!20},roundcorner=8pt]{example}
\surroundwithmdframed[hidealllines=true,backgroundcolor=\ENVCOLOR{orange!20},roundcorner=8pt]{exemple}
\surroundwithmdframed[hidealllines=true,backgroundcolor=\ENVCOLOR{orange!10},roundcorner=8pt]{remark}
\ifnotSLIDE{\surroundwithmdframed[backgroundcolor=\ENVCOLOR{green!15},roundcorner=8pt]{definition}}
\surroundwithmdframed[hidealllines=true,backgroundcolor=\ENVCOLOR{gray!35}]{proposition}
\ifnotSLIDE{\surroundwithmdframed[backgroundcolor=\ENVCOLOR{gray!35},roundcorner=8pt]{theorem}}
\surroundwithmdframed[hidealllines=true,backgroundcolor=\ENVCOLOR{gray!35}]{lemma}
\surroundwithmdframed[hidealllines=true,backgroundcolor=\ENVCOLOR{gray!35}]{corollary}

\surroundwithmdframed[backgroundcolor=red!35,roundcorner=8pt]{theoremaprouver}
\surroundwithmdframed[backgroundcolor=red!35,roundcorner=8pt]{conjectureaprouver}
}

\usepackage{tcolorbox}
\tcbuselibrary{skins}

\newtcolorbox{maboiteconfiance}[2][]{colback=red!5!white, colframe=red!75!black,fonttitle=\bfseries, colbacktitle=red!85!black,enhanced,
attach boxed title to top center={yshift=-2mm},
  title={#2},#1}
  
\newtcolorbox{maboiteresume}[2][]{colback=red!5!white, colframe=red!75!black,fonttitle=\bfseries, colbacktitle=red!85!black,enhanced,
attach boxed title to top center={yshift=-2mm},
  title={#2},#1}
  \newtcolorbox{maboiteresumeperso}[2][]{colback=red!5!white, colframe=red!75!black,fonttitle=\bfseries, colbacktitle=red!85!black,enhanced,
attach boxed title to top center={yshift=-2mm},
  title={#2},#1}
  \newtcolorbox{maboitecommentaire}[2][]{colback=red!5!white, colframe=red!75!black,fonttitle=\bfseries, colbacktitle=red!85!black,enhanced,
attach boxed title to top center={yshift=-2mm},
  title={#2},#1}
\newtcolorbox{maboitetruc}[2][]{colback=red!5!white, colframe=red!75!black,fonttitle=\bfseries, colbacktitle=red!85!black,enhanced,
attach boxed title to top left={yshift=-2mm},
  title={#2},#1}

\ifdefined\CESTDUFRANCAIS
	\providecommand\LANGUE[2]{#2}
	\usepackage[french]{babel}
\else
	\providecommand\LANGUE[2]{#1}
\fi
\providecommand\LANGUEinv[2]{\LANGUE{#2}{#1}}

\providecommand\ifnotmodeDIFFUSIONFINALE[1]{
\ifdefined\SAFEMODE
\else
#1
\fi
}
\providecommand\ifmodeDIFFUSIONFINALE[1]{
\ifdefined\SAFEMODE
#1
\else
\fi
}

\usepackage{xr-hyper}
\usepackage{hyperref}

\usepackage{versions}
\usepackage{amsmath,amssymb,mathrsfs,amsfonts}
\usepackage{listings}
\usepackage{url}
\usepackage{versions}
\usepackage{color}
\usepackage[colorinlistoftodos]{todonotes}
\usepackage{proof}
\usepackage{xstring}

\usepackage{makeidx}
\makeindex
\providecommand\DOINDEXPRINT{}
\providecommand\SIGNALEMOTNOUV[1]{}
\ifdefined\INDEXPRINT
	\ifnotSAFEMODE{
	\renewcommand\DOINDEXPRINT{
	      	\let\oldindex\index
		\renewcommand{\index}[1]{\oldindex{##1}\marginpar{IDXLA: \tiny##1}}
		\renewcommand\SIGNALEMOTNOUV[1]{\marginpar{MNV: \small\textbf{##1}}}
	}
	}
\fi

\usepackage[utf8]{inputenc}
\usepackage[T1]{fontenc}

\usepackage{bbm}

\DeclareUnicodeCharacter{2212}{-}
\DeclareUnicodeCharacter{2297}{\ensuremath{\otimes}}       
\DeclareUnicodeCharacter{2A33}{\ensuremath{\smashtimes}}
\DeclareUnicodeCharacter{1D56F}{\ensuremath{\mathfrak{D}}}  
\DeclareUnicodeCharacter{1D40E}{\ensuremath{\mathbf{O}}}  
\DeclareUnicodeCharacter{1D427}{\ensuremath{\mathbf{n}}}  
\DeclareUnicodeCharacter{3008}{\ensuremath{\langle}} 
\DeclareUnicodeCharacter{3009}{\ensuremath{\rangle}} 

\DeclareUnicodeCharacter{00AC}{\ensuremath{\neg}}                     
\DeclareUnicodeCharacter{00B1}{\ensuremath{\pm}}                      
\DeclareUnicodeCharacter{00B2}{\ensuremath{^2}}                       
\DeclareUnicodeCharacter{00B3}{\ensuremath{^3}}                       
\DeclareUnicodeCharacter{00B7}{\ensuremath{\cdot}}                    
\DeclareUnicodeCharacter{00B9}{\ensuremath{^1}}                       
\DeclareUnicodeCharacter{00D7}{\ensuremath{\times}}                   
\DeclareUnicodeCharacter{02B0}{\ensuremath{^h}}                       
\DeclareUnicodeCharacter{02B3}{\ensuremath{^r}}                       
\DeclareUnicodeCharacter{02E2}{\ensuremath{^s}}                       
\DeclareUnicodeCharacter{0302}{\ensuremath{\hat{\phantom{x}}}}        
\DeclareUnicodeCharacter{0332}{\ensuremath{\underline{\phantom{x}}}}  
\DeclareUnicodeCharacter{0308}{\ensuremath{\ddot{\phantom{x}}}}       
\DeclareUnicodeCharacter{0393}{\ensuremath{\Gamma}}                   
\DeclareUnicodeCharacter{0394}{\ensuremath{\Delta}}                   
\DeclareUnicodeCharacter{0398}{\ensuremath{\Theta}}                   
\DeclareUnicodeCharacter{039B}{\ensuremath{\Lambda}}                  
\DeclareUnicodeCharacter{039E}{\ensuremath{\Xi}}                      
\DeclareUnicodeCharacter{03A0}{\ensuremath{\Pi}}                      
\DeclareUnicodeCharacter{03A3}{\ensuremath{\Sigma}}                   

\DeclareUnicodeCharacter{03A4}{\ensuremath{\mathrm{T}}}  

\DeclareUnicodeCharacter{03A5}{\ensuremath{\Upsilon}}   
\DeclareUnicodeCharacter{03A6}{\ensuremath{\Phi}}       
\DeclareUnicodeCharacter{03A8}{\ensuremath{\Psi}}       
\DeclareUnicodeCharacter{03A9}{\ensuremath{\Omega}}     
\DeclareUnicodeCharacter{03B1}{\ensuremath{\alpha}}     
\DeclareUnicodeCharacter{03B2}{\ensuremath{\beta}}      
\DeclareUnicodeCharacter{03B3}{\ensuremath{\gamma}}     
\DeclareUnicodeCharacter{03B4}{\ensuremath{\delta}}     
\DeclareUnicodeCharacter{03B5}{\ensuremath{\epsilon}}   
\DeclareUnicodeCharacter{03B6}{\ensuremath{\zeta}}      
\DeclareUnicodeCharacter{03B7}{\ensuremath{\eta}}       
\DeclareUnicodeCharacter{03B8}{\ensuremath{\theta}}     
\DeclareUnicodeCharacter{03B9}{\ensuremath{\iota}}      
\DeclareUnicodeCharacter{03BA}{\ensuremath{\kappa}}     
\DeclareUnicodeCharacter{03BB}{\ensuremath{\lambda}}    
\DeclareUnicodeCharacter{03BC}{\ensuremath{\mu}}        
\DeclareUnicodeCharacter{03BD}{\ensuremath{\nu}}        
\DeclareUnicodeCharacter{03BE}{\ensuremath{\xi}}        
\DeclareUnicodeCharacter{03C0}{\ensuremath{\pi}}        
\DeclareUnicodeCharacter{03C1}{\ensuremath{\rho}}       
\DeclareUnicodeCharacter{03C2}{\ensuremath{\varsigma}}  
\DeclareUnicodeCharacter{03C3}{\ensuremath{\sigma}}     
\DeclareUnicodeCharacter{03C4}{\ensuremath{\tau}}       
\DeclareUnicodeCharacter{03C5}{\ensuremath{\upsilon}}   
\DeclareUnicodeCharacter{03C6}{\ensuremath{\varphi}}    
\DeclareUnicodeCharacter{03C7}{\ensuremath{\chi}}       
\DeclareUnicodeCharacter{03C8}{\ensuremath{\psi}}       
\DeclareUnicodeCharacter{03C9}{\ensuremath{\omega}}     
\DeclareUnicodeCharacter{03D1}{\ensuremath{\vartheta}}  

\DeclareUnicodeCharacter{0432}{\ensuremath{\mathrm{PDF\;TODO}}}  
\DeclareUnicodeCharacter{0435}{\ensuremath{\mathrm{PDF\;TODO}}}  
\DeclareUnicodeCharacter{0438}{\ensuremath{\mathrm{PDF\;TODO}}}  
\DeclareUnicodeCharacter{043C}{\ensuremath{\mathrm{PDF\;TODO}}}  
\DeclareUnicodeCharacter{0440}{\ensuremath{\mathrm{PDF\;TODO}}}  
\DeclareUnicodeCharacter{0442}{\ensuremath{\mathrm{PDF\;TODO}}}  
\DeclareUnicodeCharacter{041F}{\ensuremath{\mathrm{PDF\;TODO}}}  
\DeclareUnicodeCharacter{0BA8}{\ensuremath{\mathrm{PDF\;TODO}}}  
\DeclareUnicodeCharacter{0BBF}{\ensuremath{\mathrm{PDF\;TODO}}}  
\DeclareUnicodeCharacter{1100}{\ensuremath{\mathrm{PDF\;TODO}}}  
\DeclareUnicodeCharacter{11F9}{\ensuremath{\mathrm{PDF\;TODO}}}  

\DeclareUnicodeCharacter{1D48}{\ensuremath{^d}}  
\DeclareUnicodeCharacter{1D50}{\ensuremath{^m}}  
\DeclareUnicodeCharacter{1D56}{\ensuremath{^p}}  
\DeclareUnicodeCharacter{1D62}{\ensuremath{_i}}  
\DeclareUnicodeCharacter{1D63}{\ensuremath{_r}}  
\DeclareUnicodeCharacter{1D64}{\ensuremath{_u}}  
\DeclareUnicodeCharacter{1D9C}{\ensuremath{^c}}  

\DeclareUnicodeCharacter{2032}{\ensuremath{\text{\textasciiacute}}}  

\DeclareUnicodeCharacter{2070}{\ensuremath{^0}}  
\DeclareUnicodeCharacter{2074}{\ensuremath{^4}}  
\DeclareUnicodeCharacter{2075}{\ensuremath{^5}}  
\DeclareUnicodeCharacter{2076}{\ensuremath{^6}}  
\DeclareUnicodeCharacter{2077}{\ensuremath{^7}}  
\DeclareUnicodeCharacter{2078}{\ensuremath{^8}}  
\DeclareUnicodeCharacter{2079}{\ensuremath{^9}}  
\DeclareUnicodeCharacter{207A}{\ensuremath{^+}}  
\DeclareUnicodeCharacter{207B}{\ensuremath{^-}}  
\DeclareUnicodeCharacter{207F}{\ensuremath{^n}}  
\DeclareUnicodeCharacter{2080}{\ensuremath{_0}}  
\DeclareUnicodeCharacter{2081}{\ensuremath{_1}}  
\DeclareUnicodeCharacter{2082}{\ensuremath{_2}}  
\DeclareUnicodeCharacter{2083}{\ensuremath{_3}}  
\DeclareUnicodeCharacter{2084}{\ensuremath{_4}}  
\DeclareUnicodeCharacter{2085}{\ensuremath{_5}}  
\DeclareUnicodeCharacter{2086}{\ensuremath{_6}}  
\DeclareUnicodeCharacter{2087}{\ensuremath{_7}}  
\DeclareUnicodeCharacter{2088}{\ensuremath{_8}}  
\DeclareUnicodeCharacter{2089}{\ensuremath{_9}}  
\DeclareUnicodeCharacter{208A}{\ensuremath{_+}}  
\DeclareUnicodeCharacter{208B}{\ensuremath{_-}}  
\DeclareUnicodeCharacter{2091}{\ensuremath{_e}}  
\DeclareUnicodeCharacter{2095}{\ensuremath{_h}}  
\DeclareUnicodeCharacter{2096}{\ensuremath{_k}}  
\DeclareUnicodeCharacter{2097}{\ensuremath{_l}}  
\DeclareUnicodeCharacter{2098}{\ensuremath{_m}}  
\DeclareUnicodeCharacter{2099}{\ensuremath{_n}}  
\DeclareUnicodeCharacter{209A}{\ensuremath{_p}}  
\DeclareUnicodeCharacter{209B}{\ensuremath{_s}}  
\DeclareUnicodeCharacter{209C}{\ensuremath{_t}}  

\DeclareUnicodeCharacter{2113}{\ensuremath{\mathpzc{l}}} 

\DeclareUnicodeCharacter{2115}{\ensuremath{\mathbbm{N}}}  
\DeclareUnicodeCharacter{211A}{\ensuremath{\mathbbm{Q}}}  
\DeclareUnicodeCharacter{211D}{\ensuremath{\mathbbm{R}}}  
\DeclareUnicodeCharacter{2124}{\ensuremath{\mathbbm{Z}}}  

\DeclareUnicodeCharacter{2135}{\ensuremath{\aleph}}           
\DeclareUnicodeCharacter{2190}{\ensuremath{\leftarrow}}       
\DeclareUnicodeCharacter{2192}{\ensuremath{\rightarrow}}      
\DeclareUnicodeCharacter{2194}{\ensuremath{\leftrightarrow}}  

\DeclareUnicodeCharacter{21A0}{\ensuremath{\twoheadrightarrow}}  

\DeclareUnicodeCharacter{21A6}{\ensuremath{\mapsto}}  
\DeclareUnicodeCharacter{21A6}{\ensuremath{\mapsto}}  

\DeclareUnicodeCharacter{21BE}{\ensuremath{\restriction}}  

\DeclareUnicodeCharacter{21D0}{\ensuremath{\Leftarrow}}       
\DeclareUnicodeCharacter{21D2}{\ensuremath{\Rightarrow}}      
\DeclareUnicodeCharacter{21D4}{\ensuremath{\Leftrightarrow}}  

\DeclareUnicodeCharacter{2200}{\ensuremath{\forall}}      
\DeclareUnicodeCharacter{2203}{\ensuremath{\exists}}      
\DeclareUnicodeCharacter{2204}{\ensuremath{\not\exists}}  
\DeclareUnicodeCharacter{2205}{\ensuremath{\emptyset}}    
\DeclareUnicodeCharacter{2208}{\ensuremath{\in}}          
\DeclareUnicodeCharacter{2209}{\ensuremath{\not\in}}      
\DeclareUnicodeCharacter{220B}{\ensuremath{\ni}}          
\DeclareUnicodeCharacter{220C}{\ensuremath{\not\ni}}      

\DeclareUnicodeCharacter{220E}{{\tiny \ensuremath{\blacksquare}}} 

\DeclareUnicodeCharacter{220F}{{\tiny \ensuremath{\prod}}}  
\DeclareUnicodeCharacter{2211}{\ensuremath{\sum}}           
\DeclareUnicodeCharacter{2218}{\ensuremath{\circ}}          
\DeclareUnicodeCharacter{2219}{\ensuremath{\bullet}}        
\DeclareUnicodeCharacter{221E}{\ensuremath{\infty}}         
\DeclareUnicodeCharacter{2223}{\ensuremath{\mid}}           
\DeclareUnicodeCharacter{2225}{\ensuremath{\parallel}}      
\DeclareUnicodeCharacter{2227}{\ensuremath{\wedge}}         
\DeclareUnicodeCharacter{2228}{\ensuremath{\vee}}           
\DeclareUnicodeCharacter{2229}{\ensuremath{\cap}}           
\DeclareUnicodeCharacter{222A}{\ensuremath{\cup}}           
\DeclareUnicodeCharacter{222B}{\ensuremath{\int}}           
\DeclareUnicodeCharacter{2234}{\ensuremath{\therefore}}     

\DeclareUnicodeCharacter{2237}{\ensuremath{\dblcolon}} 

\newcommand{\dotminus}{\mathop{\mbox{$-^{\hspace{-.5em}\cdot}\,$}}}
\DeclareUnicodeCharacter{2238}{\ensuremath{\dotminus}}  

\DeclareUnicodeCharacter{223C}{\ensuremath{\sim}}       
\DeclareUnicodeCharacter{2243}{\ensuremath{\simeq}}     
\DeclareUnicodeCharacter{2245}{\ensuremath{\cong}}      
\DeclareUnicodeCharacter{2248}{\ensuremath{\approx}}    
\DeclareUnicodeCharacter{224D}{\ensuremath{\asymp}}       
\DeclareUnicodeCharacter{2250}{\ensuremath{\doteq}}     

\DeclareUnicodeCharacter{2254}{\ensuremath{\coloneqq}}  

\newcommand{\eqq}{\stackrel{\text{\tiny ?}}{=}}
\DeclareUnicodeCharacter{225F}{\ensuremath{\eqq}}  

\DeclareUnicodeCharacter{2260}{\ensuremath{\neq}}         
\DeclareUnicodeCharacter{2261}{\ensuremath{\equiv}}       
\DeclareUnicodeCharacter{2262}{\ensuremath{\not\equiv}}   
\DeclareUnicodeCharacter{2264}{\ensuremath{\le}}          
\DeclareUnicodeCharacter{2265}{\ensuremath{\ge}}          
\DeclareUnicodeCharacter{226E}{\ensuremath{\nless}}       
\DeclareUnicodeCharacter{226F}{\ensuremath{\ngtr}}        
\DeclareUnicodeCharacter{2270}{\ensuremath{\nleq}}        
\DeclareUnicodeCharacter{2271}{\ensuremath{\ngeq}}        
\DeclareUnicodeCharacter{227A}{\ensuremath{\prec}}        
\DeclareUnicodeCharacter{227C}{\ensuremath{\preceq}}      
\DeclareUnicodeCharacter{2282}{\ensuremath{\subset}}      
\DeclareUnicodeCharacter{2284}{\ensuremath{\not\subset}} 
\DeclareUnicodeCharacter{2283}{\ensuremath{\supset}}      
\DeclareUnicodeCharacter{2286}{\ensuremath{\subseteq}}    
\DeclareUnicodeCharacter{228E}{\ensuremath{\uplus}}       
\DeclareUnicodeCharacter{2291}{\ensuremath{\sqsubseteq}}  
\DeclareUnicodeCharacter{2293}{\ensuremath{\sqcap}}       
\DeclareUnicodeCharacter{2294}{\ensuremath{\sqcup}}       
\DeclareUnicodeCharacter{2295}{\ensuremath{\oplus}}       
\DeclareUnicodeCharacter{22A2}{\ensuremath{\vdash}}       
\DeclareUnicodeCharacter{22A4}{\ensuremath{\top}}         
\DeclareUnicodeCharacter{22A5}{\ensuremath{\bot}}         
\DeclareUnicodeCharacter{22C0}{\ensuremath{\bigwedge}}    
\DeclareUnicodeCharacter{22C2}{\ensuremath{\bigcap}}      
\DeclareUnicodeCharacter{22C3}{\ensuremath{\bigcup}}      
\DeclareUnicodeCharacter{22EE}{\ensuremath{\vdots}}       
\DeclareUnicodeCharacter{22EF}{\ensuremath{\cdots}}       

\DeclareUnicodeCharacter{22F0}{\ensuremath{\iddots}} 

\DeclareUnicodeCharacter{230A}{\ensuremath{\lfloor}}  
\DeclareUnicodeCharacter{230B}{\ensuremath{\rfloor}}  
\DeclareUnicodeCharacter{2308}{\ensuremath{\lceil}}   
\DeclareUnicodeCharacter{2309}{\ensuremath{\rceil}}   
\DeclareUnicodeCharacter{2329}{\ensuremath{\langle}}  
\DeclareUnicodeCharacter{232A}{\ensuremath{\rangle}}  

\DeclareUnicodeCharacter{23CE}{\ensuremath{\Return}}  

\DeclareUnicodeCharacter{2460}{\ensuremath{\text{\ding{172}}}} 
\DeclareUnicodeCharacter{2461}{\ensuremath{\text{\ding{173}}}} 
\DeclareUnicodeCharacter{2462}{\ensuremath{\text{\ding{174}}}} 
\DeclareUnicodeCharacter{2463}{\ensuremath{\text{\ding{175}}}} 
\DeclareUnicodeCharacter{2464}{\ensuremath{\text{\ding{176}}}} 
\DeclareUnicodeCharacter{2465}{\ensuremath{\text{\ding{177}}}} 
\DeclareUnicodeCharacter{2466}{\ensuremath{\text{\ding{178}}}} 
\DeclareUnicodeCharacter{2467}{\ensuremath{\text{\ding{179}}}} 
\DeclareUnicodeCharacter{2468}{\ensuremath{\text{\ding{180}}}} 

\DeclareUnicodeCharacter{25C5}{\ensuremath{\triangleleft}}  

\DeclareUnicodeCharacter{2660}{\ensuremath{\spadesuit}} 
\DeclareUnicodeCharacter{266D}{\ensuremath{\flat}}      
\DeclareUnicodeCharacter{266F}{\ensuremath{\sharp}}     

\DeclareUnicodeCharacter{2713}{\ensuremath{\checkmark}}  

\DeclareUnicodeCharacter{27E6}{\ensuremath{\llbracket}}  
\DeclareUnicodeCharacter{27E7}{\ensuremath{\rrbracket}}  

\DeclareUnicodeCharacter{27E8}{\ensuremath{\langle}}          
\DeclareUnicodeCharacter{27E9}{\ensuremath{\rangle}}          
\DeclareUnicodeCharacter{27F6}{\ensuremath{\longrightarrow}}  

\DeclareUnicodeCharacter{2983}{\ensuremath{\mathrm{PDF\;TODO}}}  
\DeclareUnicodeCharacter{2984}{\ensuremath{\mathrm{PDF\;TODO}}}  
\DeclareUnicodeCharacter{2985}{\ensuremath{\mathrm{PDF\;TODO}}}  
\DeclareUnicodeCharacter{2986}{\ensuremath{\mathrm{PDF\;TODO}}}  

\DeclareUnicodeCharacter{2987}{\ensuremath{\llparenthesis}}  
\DeclareUnicodeCharacter{2988}{\ensuremath{\rrparenthesis}}  

\DeclareUnicodeCharacter{29F5}{\ensuremath{\setminus}}   

\DeclareUnicodeCharacter{2A06}{\ensuremath{\bigcup}}  
\DeclareUnicodeCharacter{2C7C}{\ensuremath{_j}}       

\DeclareUnicodeCharacter{D7B0}{\ensuremath{\mathrm{PDF\;TODO}}} 

\DeclareUnicodeCharacter{1D7D9}{\ensuremath{\mathbbm{1}}}  

\DeclareUnicodeCharacter{22AC}{\ensuremath{\neg\vdash}}
\DeclareUnicodeCharacter{279C}{\ensuremath{\to}}
\DeclareUnicodeCharacter{1D442}{ }
\DeclareUnicodeCharacter{2000}{ }
\DeclareUnicodeCharacter{1D440}{M}
\DeclareUnicodeCharacter{1D442}{O}
\DeclareUnicodeCharacter{20E3}{2}
\DeclareUnicodeCharacter{1D45A}{m}
\DeclareUnicodeCharacter{1D45B}{n}

\usepackage{pifont}
\DeclareUnicodeCharacter{2714}{\ding{52}}
\DeclareUnicodeCharacter{2705}{\ding{52}}
\DeclareUnicodeCharacter{2717}{\ding{55}}
\DeclareUnicodeCharacter{274C}{\ding{53}}

\DeclareUnicodeCharacter{1F9E9}{}	
\DeclareUnicodeCharacter{2699}{}
\DeclareUnicodeCharacter{FE0F}{}
\DeclareUnicodeCharacter{2080}{\ensuremath{_{0}}}
\DeclareUnicodeCharacter{1F4DA}{}
\DeclareUnicodeCharacter{1F9EE}{}
\DeclareUnicodeCharacter{1F9E0}{}
\DeclareUnicodeCharacter{27A1}{0}
\DeclareUnicodeCharacter{1F9ED}{} 
\DeclareUnicodeCharacter{1F501}{}
\DeclareUnicodeCharacter{1F50D}{}
\DeclareUnicodeCharacter{1F4A1}{}
\DeclareUnicodeCharacter{1F4D8}{}
\DeclareUnicodeCharacter{1F447}{}
\DeclareUnicodeCharacter{1F4BB}{}
\DeclareUnicodeCharacter{1F4AC}{}
\DeclareUnicodeCharacter{1F449}{}
\DeclareUnicodeCharacter{1F4C9}{}
\DeclareUnicodeCharacter{1F522}{}
\DeclareUnicodeCharacter{1F52C}{}
\DeclareUnicodeCharacter{1F51A}{}
\DeclareUnicodeCharacter{1F33F}{}
\DeclareUnicodeCharacter{1F539}{}
\DeclareUnicodeCharacter{1F538}{}
\DeclareUnicodeCharacter{1F44C}{}
\DeclareUnicodeCharacter{1F7E6}{}
\DeclareUnicodeCharacter{1F44F}{}
\DeclareUnicodeCharacter{26A0}{}
\DeclareUnicodeCharacter{1F517}{}
\DeclareUnicodeCharacter{1F52D}{}
\DeclareUnicodeCharacter{1F44D}{}
\DeclareUnicodeCharacter{202F}{}
\DeclareUnicodeCharacter{1F3AF}{}
\DeclareUnicodeCharacter{1F4CF}{}

\usepackage{stmaryrd}
\usepackage{ifthen}
\usepackage{mathtools}
\usepackage{dsfont}
\PASSIAPXPROOF{\usepackage{thmtools}}
\usepackage{longtable}
\usepackage{mathdots}
\usepackage{booktabs}

\usepackage{tikz}
\usetikzlibrary{calc}
\usepgflibrary{arrows}
\usetikzlibrary{fit}
\usetikzlibrary{patterns}
\usetikzlibrary{decorations.pathreplacing}
\usetikzlibrary{shapes.multipart}
\usetikzlibrary{shapes.geometric}
\usetikzlibrary{shapes.misc}
\usetikzlibrary{positioning}
\usetikzlibrary{graphs}
\usetikzlibrary{topaths}
\usetikzlibrary{arrows.meta}
\usetikzlibrary{decorations.pathreplacing}
\usetikzlibrary{decorations.markings}
\usetikzlibrary{decorations.pathmorphing}
 \usetikzlibrary{calc}
 \usetikzlibrary{automata}
\usetikzlibrary{chains}
\usetikzlibrary{scopes}
 \usetikzlibrary{positioning}
 \usetikzlibrary{through,backgrounds}
\usepackage{circuitikz}

\providecommand\ifnotPOLYCHAPTERMODE[1]{#1}
\providecommand\ifPOLYCHAPTERMODE[1]{}

\ifnotPOLYCHAPTERMODE{
\ifdefined\thechapter
\providecommand\spnewtheorem[4]{\newtheorem{#1}{#2}[chapter]}
\providecommand\spnewtheoremi[5]{\newtheorem{#1}[#5]{#2}}
\else
\providecommand\spnewtheorem[4]{\newtheorem{#1}{#2}}
\providecommand\spnewtheoremi[5]{\newtheorem{#1}[#5]{#2}}
\fi
}
\ifPOLYCHAPTERMODE{
\providecommand\spnewtheorem[4]{\newtheorem{#1}{#2}}
\providecommand\spnewtheoremi[5]{\newtheorem{#1}[#5]{#2}}
}
\providecommand\nolipics[1]{#1}

\includeversion{metdate}

\newcommand\XPAGEDEGARDE[2]{
\thispagestyle{empty}
\title{{#1} \\[2cm] 
{\large #2}
}
\author{
 Olivier Bournez
\\[0.5cm] 
bournez@lix.polytechnique.fr
}
\begin{metdate}
\date{
\vspace{2cm}
Version \LANGUE{of}{du} \today \\[1cm] 
\includegraphics[height=3.5cm]{/Users/bournez/lib/LaTeX/Perso/FIG-COMMONS/logo_X_new}
}
\end{metdate}
\maketitle
}

\newcommand\vectorl[1]{{\mathbf#1}}

\newcommand\vx{\vectorl{x}}

\newcommand\Q{\mathbb{Q}}
\newcommand\R{\mathbb{R}}
\newcommand\Z{\mathbb{Z}}

\newcommand\N{\mathbb{N}}

\newcommand{\cp}[1]{\operatorname{#1}}

\LANGUE{

}
{
}

\newcommand\NP{\cp{NP}}

\LANGUEinv{

}{

}

\LANGUEinv{}
{}

\LANGUEinv{}
{}

\LANGUEinv{}
{}
\LANGUEinv{}{
}
\LANGUEinv{
}
{}

\LANGUEinv{
}
{}
\LANGUEinv{
}
{}
\LANGUEinv{
}
{}
\LANGUEinv{
}
{}
\LANGUEinv{
}
{}
\LANGUEinv{
}
{}

\LANGUEinv{
}
{}
\LANGUEinv{
}
{}

\LANGUEinv{
}
{}

\LANGUEinv{
}
{}

\newcommand\semspace{\kern -.25ex}

\renewcommand{\restriction}{\mathord{\upharpoonright}}

\newcommand{\myop}[1]{\operatorname{#1}}

\newcommand{\myclass}[1]{\operatorname{#1}}
 \newcommand{\gval}[2][]{\ensuremath{\myclass{GVAL}_{#1}\ifthenelse{\equal{#2}{}}{}{[#2]}}}

 \newcommand{\gpc}[1][]{\ensuremath{\myclass{ATSP}}}
\newcommand{\cgc}[1][]{\ensuremath{\myclass{ATS}}}
\newcommand{\gexpc}[1][]{\ensuremath{\myclass{AEXP}}}

\newcommand{\gplc}[1][]{\ensuremath{\myclass{ALP}}}
\newcommand{\cglc}[1][]{\ensuremath{\myclass{AL}}}

\newenvironment{dessinp}[1]{
\begin{center}\begin{tikzpicture}[baseline=(current bounding
      box.north),#1]
}
{\end{tikzpicture}
\end{center}}

\newcommand{\round}[1]{\left\lfloor#1\right\rceil}

\newcommand{\MAT}[3]{M_{#1\ifthenelse{\equal{#2}{}}{}{,#2}}\ifthenelse{\equal{#3}{}}{}{\left(#3\right)}}

\ifnotSLIDE{
\spnewtheorem{openproblem}{Open Problem}{\bfseries}{\rmfamily}
\spnewtheorem{remarkolivier}{Commentaire d'Olivier: }{\bfseries}{\rmfamily}
\spnewtheorem{postulate}{Postulate}{\bfseries}{\rmfamily}
}

\ifnotSLIDE{
\nolipics{
\spnewtheorem{theorem}{\LANGUE{Theorem}{Théorème}}{\bfseries}{\rmfamily}
\spnewtheoremi{definition}{\LANGUE{Definition}{Définition}}{\bfseries}{\rmfamily}{theorem} 
\spnewtheoremi{lemma}{\LANGUE{Lemma}{Lemme}}{\bfseries}{\rmfamily}{theorem} 
\spnewtheoremi{corollary}{\LANGUE{Corollary}{Corollaire}}{\bfseries}{\rmfamily}{theorem} 
\spnewtheoremi{proposition}{Proposition}{\bfseries}{\rmfamily}{theorem} 
\spnewtheoremi{example}{\LANGUE{Example}{Exemple}}{\bfseries}{\rmfamily}{theorem} 
\spnewtheoremi{remark}{\LANGUE{Remark}{Remarque}}{\bfseries}{\rmfamily}{theorem} 
\spnewtheorem{exercise}{\LANGUE{Exercise}{Exercice}}{\bfseries}{\rmfamily}
\spnewtheorem{exercisee}{\LANGUE{*Exercise}{*Exercice}}{\bfseries}{\rmfamily}
\spnewtheorem{summary}{\LANGUE{Summary}{Résumé}}{\bfseries}{\rmfamily}
 \spnewtheorem{conjecture}{\LANGUE{Conjecture}{Conjecture}}{\bfseries}{\rmfamily}
 \spnewtheorem{theoremaprouver}{\LANGUE{Theorem TO BE PROVED}{Théorème (A PROUVER)}}{\bfseries}{\rmfamily}
  \spnewtheorem{conjectureaprouver}{\LANGUE{Conjecture TO BE VERIFIED}{Conjecture (A VERIFIER)}}{\bfseries}{\rmfamily}
 }}
            
 \makeatletter
 \ifnotSLIDE{
  \ifnotTDEXAM{
\@ifundefined{proof}{
	\newenvironment{proof}{{\bf \LANGUE{Proof}{Démonstration}}:}{\hfill $\Box$ 
	}{}
	}}}
\makeatother

\newcommand\zero{\vectorl{0}}
\newcommand\un{\vectorl{1}}

\providecommand\citep[1]{\cite{#1}}

\renewcommand\Pr{\cp{Pr}}

\tikzstyle{case} = []

\RequirePackage{listings}
\def\beginlisting{\begin{lstlisting}}
\def\endlisting{\end{lstlisting}}

\newenvironment{algo}[1]{
  \lstset{escapechar=\@}
  \def\input{$#1$}

}{}

\newenvironment{FRANCAIS}{
\small
\color{orange}
}
{\normalsize
\color{black}
}

\newcommand\nl{\ \\ }

\usepackage{csquotes}

\definecolor{deepgreen}{RGB}{0,160,40}
\definecolor{forestgreen}{RGB}{0,150,70}
\definecolor{mossgreen}{RGB}{0,140,90}
\definecolor{greenwave}{RGB}{0,150,90}
\definecolor{forestmist}{RGB}{0,160,90}
\definecolor{darkteal}{RGB}{0,130,70}

\definecolor{seagreen}{RGB}{0,130,100}
\definecolor{freshteal}{RGB}{0,140,130}
\definecolor{greenteal}{RGB}{0,140,120}
\definecolor{blueteal}{RGB}{0,150,120}      
\definecolor{lagoon}{RGB}{0,180,120}
\definecolor{greenmist}{RGB}{0,170,110}

\definecolor{springgreen}{RGB}{0,190,100}
\definecolor{mintgreen}{RGB}{10,200,150}
\definecolor{oceanfoam}{RGB}{0,200,150}
\definecolor{aquagreen}{RGB}{30,210,170}
\definecolor{coastalaqua}{RGB}{40,210,180}
\definecolor{softaqua}{RGB}{40,210,190}

\definecolor{emeraldteal}{RGB}{0,180,140}
\definecolor{deepcyan}{RGB}{0,180,160}
\definecolor{lightcyan}{RGB}{0,200,160}
\definecolor{aquamarine}{RGB}{0,200,160}
\definecolor{lightaqua}{RGB}{60,230,200}
\definecolor{mistblue}{RGB}{80,240,210}

\definecolor{paleaqua}{RGB}{80,200,190}
\definecolor{glacier}{RGB}{90,150,200}
\definecolor{skyblueA}{RGB}{60,150,220}      
\definecolor{lightsky}{RGB}{90,170,230}
\definecolor{iceblue}{RGB}{120,190,240}
\definecolor{frost}{RGB}{150,210,250}

\definecolor{steelblue}{RGB}{30,100,160}
\definecolor{cadetblue}{RGB}{70,120,180}
\definecolor{coolblue}{RGB}{110,170,220}
\definecolor{morningblue}{RGB}{140,200,240}
\definecolor{blueiceA}{RGB}{40,150,210}       
\definecolor{blueair}{RGB}{60,170,225}

\definecolor{tealblue}{RGB}{0,110,130}
\definecolor{deepbluecyan}{RGB}{0,110,170}
\definecolor{azur}{RGB}{0,70,180}
\definecolor{brightazure}{RGB}{20,130,200}
\definecolor{lightazure}{RGB}{80,180,235}
\definecolor{skylight}{RGB}{100,200,245}

\definecolor{nightblue}{RGB}{40,0,110}
\definecolor{slateblue}{RGB}{20,40,120}
\definecolor{violetA}{RGB}{60,50,160}         
\definecolor{blueviolet}{RGB}{80,60,180}
\definecolor{purpleA}{RGB}{100,70,200}        
\definecolor{lightpurple}{RGB}{130,90,210}

\definecolor{lilac}{RGB}{160,110,220}
\definecolor{indigoA}{RGB}{50,0,130}          
\definecolor{deepindigo}{RGB}{70,10,150}
\definecolor{royalpurple}{RGB}{90,20,170}

\definecolor{coolmint}{RGB}{0,130,70}
\definecolor{deepsea}{RGB}{0,130,70}

\definecolor{crimson}{RGB}{200,30,60}
\definecolor{brickred}{RGB}{180,40,40}
\definecolor{coral}{RGB}{220,90,70}
\definecolor{salmon}{RGB}{240,130,120}

\definecolor{magenta}{RGB}{200,40,150}
\definecolor{fuchsia}{RGB}{220,60,170}
\definecolor{violetB}{RGB}{140,70,190}        
\definecolor{indigoB}{RGB}{90,50,180}         

\definecolor{royalblueA}{RGB}{50,80,200}      
\definecolor{azure}{RGB}{30,120,230}
\definecolor{skyblueB}{RGB}{80,160,230}       
\definecolor{tealB}{RGB}{0,150,170}           

\definecolor{marine}{RGB}{0,80,140}

\newcommand{\ProtaColorPair}[1]{%
  \colorlet{#1fill}{#1}%
  \colorlet{#1text}{#1!80!black}%
}

	\foreach \i/\name in {
	  1/deepgreen,
	  2/greenwave,
	  3/lagoon,
	  4/oceanfoam,
	  5/softaqua,
	  6/mistblue,
 	 7/lightsky,
	  8/blueiceA,   
	  9/deepcyan,
	  10/azur,
	  11/blueviolet,
	  12/royalpurple}
	  {\ProtaColorPair{\name}}

\newcommand{\DarkenForText}[2]{%
  \colorlet{#1text}{black!#2!#1}%
}

\DarkenForText{deepgreen}{20}
\DarkenForText{forestgreen}{20}
\DarkenForText{mossgreen}{20}
\DarkenForText{greenwave}{20}
\DarkenForText{forestmist}{20}
\DarkenForText{darkteal}{20}

\DarkenForText{seagreen}{20}
\DarkenForText{freshteal}{20}
\DarkenForText{greenteal}{20}
\DarkenForText{blueteal}{20}
\DarkenForText{lagoon}{20}
\DarkenForText{greenmist}{20}

\DarkenForText{springgreen}{20}
\DarkenForText{mintgreen}{20}
\DarkenForText{oceanfoam}{20}
\DarkenForText{aquagreen}{20}
\DarkenForText{coastalaqua}{20}
\DarkenForText{softaqua}{20}

\DarkenForText{emeraldteal}{20}
\DarkenForText{deepcyan}{20}
\DarkenForText{lightcyan}{20}
\DarkenForText{aquamarine}{20}
\DarkenForText{lightaqua}{20}
\DarkenForText{mistblue}{30}

\DarkenForText{paleaqua}{20}
\DarkenForText{glacier}{20}
\DarkenForText{skyblueA}{20}
\DarkenForText{lightsky}{20}
\DarkenForText{iceblue}{20}
\DarkenForText{frost}{20}

\DarkenForText{steelblue}{20}
\DarkenForText{cadetblue}{20}
\DarkenForText{coolblue}{20}
\DarkenForText{morningblue}{20}
\DarkenForText{blueiceA}{20}
\DarkenForText{blueair}{20}

\DarkenForText{tealblue}{20}
\DarkenForText{deepbluecyan}{20}
\DarkenForText{azur}{20}
\DarkenForText{brightazure}{20}
\DarkenForText{lightazure}{20}
\DarkenForText{skylight}{20}

\DarkenForText{nightblue}{20}
\DarkenForText{slateblue}{20}
\DarkenForText{violetA}{20}
\DarkenForText{blueviolet}{20}
\DarkenForText{purpleA}{20}
\DarkenForText{lightpurple}{20}

\DarkenForText{lilac}{20}
\DarkenForText{indigoA}{20}
\DarkenForText{deepindigo}{20}
\DarkenForText{royalpurple}{20}

\DarkenForText{coolmint}{20}
\DarkenForText{deepsea}{20}

\DarkenForText{crimson}{20}
\DarkenForText{brickred}{20}
\DarkenForText{coral}{20}
\DarkenForText{salmon}{20}

\DarkenForText{magenta}{20}
\DarkenForText{fuchsia}{20}
\DarkenForText{violetB}{20}
\DarkenForText{indigoB}{20}

\DarkenForText{royalblueA}{20}
\DarkenForText{azure}{20}
\DarkenForText{skyblueB}{20}
\DarkenForText{tealB}{20}

\DarkenForText{marine}{20}

\providecommand\apxparameter{}

\apxproofmode{
\usepackage[\apxparameter]{apxproof}
\theoremstyle{plain}

\newtheoremrep{theorem}{Theorem}
\newtheoremrep{claim}[theorem]{Theorem}
\newtheoremrep{proposition}[theorem]{Proposition}
\newtheoremrep{lemma}[theorem]{Lemma}
\newtheoremrep{corollary}[theorem]{Corollary}
\newtheoremrep{resume}[theorem]{Summary}
}

\makeatletter
\apxproofmodesubsection{
\let\axp@section\subsection

}
\makeatother

\makeatletter
\newcommand\appendixapxproofsubsection{
\appendix
\let\apx@orig@appendix\appendix
\renewcommand{\appendix}{}
}
\makeatother

\usepackage{xcolor}
\definecolor{check}{RGB}{160,160,0}

\includeversion{NOTLFMI2020}
\includeversion{NOTCSE304}
\excludeversion{NOTDIFFUSIONFINALE}

\providecommand\ifnotmodeDIFFUSIONFINALE[1]{
\ifdefined\DIFFUSIONFINALE
\else
#1
\fi
}

\providecommand\ifmodeDIFFUSIONFINALE[1]{
\ifdefined\DIFFUSIONFINALE
\ifdefined\SAFEMODE
\else
#1
\fi
\fi
}

\ifnotmodeDIFFUSIONFINALE{\includeversion{NOTDIFFUSIONFINALE}}
\ifnotmodeDIFFUSIONFINALE{\newcommand\FINPROVISOIRELFMI[1]{}}

\providecommand\FINPROVISOIRELFMI{
  \section{Bibliographic notes}

  The current text is highly based (sometimes copied) from:
  
  \bibliographystyle{alpha}
  \bibliography{/Users/bournez/bibliographie/Bib-Files/bournez-avantbib2DOI,/Users/bournez/bibliographie/Bib-Files/perso}
  \end{document}
  }

\providecommand\ifnotPOLYCHAPTERMODE[1]{
\ifdefined\POLYCHAPTERMODE
\else
#1
\fi
}

\newcommand\MYCOURSEON[2]{ 
	\ifnotPOLYMODE{
		\ifdefined\MAINAUXFILE
			\externaldocument{\MAINAUXFILE}
		\fi
		}
	
	\ifdefined\PASLAPREMIEREFOIS
		\newpage
	\else
		\makeindex

     		\begin{document}
	\fi
	\gdef\PASLAPREMIEREFOIS{}

	\ifdefined\XPOLY
		\XPAGEDEGARDE{\TITRECOURS\ifPOLYCHAPTERMODE{\\[0.5cm] \Large \bf \LANGUE{Chapter:}{Chapitre:}  #2}}{\SUBTITRECOURS}
		\ifPOLYCHAPTERMODE{
		\renewcommand\thesection{\arabic{section}}
		\chapter*{#2}
		}
	\else
	
      \title{ \ifnotmodeDIFFUSIONFINALE{Course  #1 -}
      \ifPOLYCHAPTERMODE{\LANGUE{Chapter:}{Chapitre:} }
      #2}
      \author{
        Olivier Bournez
        \\[0.5cm] 
      }
      \date{Version of \today}

      \ifnotPOLYCHAPTERMODE{
      	\maketitle
		}
      \ifPOLYCHAPTERMODE{\ifdefined\DOPAGEDEGARDE
      			\maketitle
		\fi
		\renewcommand\thesection{\arabic{section}}
      		\chapter*{#2}}

      \fi

      \DOINDEXPRINT
      
      \ifmodeDIFFUSIONFINALE{
		\excludeversion{recherche}
	}

    }

\newcommand\MYCOURSE[1]{
  \MYCOURSEON{\csname numerocours#1\endcsname}{\csname nomcours#1\endcsname}
  
    \ifnotmodeDIFFUSIONFINALE{}

  }
  
  \newcommand\MYCOURSEnew[1]{
  \MYCOURSEON{???}{#1}
  
  \ifnotmodeDIFFUSIONFINALE{}

  }

\newcommand\PASPREMIEREPAGE[1]{}
\providecommand\mychapter[1]{
  \renewcommand\PASPREMIEREPAGE[1]{##1}  \hrule\hrule \section*{#1}
  \hrule\hrule \  \\[1cm]}
\providecommand\chapter[1]{\PASPREMIEREPAGE{\newpage} 
  \renewcommand\PASPREMIEREPAGE[1]{##1}  \hrule\hrule \section*{#1}
  \hrule\hrule \  \\[1cm]}

\providecommand\IFNOTSUB[1]{#1}
\IFNOTSUB{

\usepackage{versions}
\includeversion{pasweb}

}

\newcommand\COURS[3]{
  \expandafter\newcommand\csname numerocours#2\endcsname{#1}
  \expandafter\newcommand\csname nomcours#2\endcsname{#3}
}

\newcommand\NANCY[3]{
  \expandafter\newcommand\csname numerocours#2\endcsname{#1}
  \expandafter\newcommand\csname nomcours#2\endcsname{#3}
}

\newcommand\INFquatrecentdouze[3]{
  \expandafter\newcommand\csname numerocours#2\endcsname{#1}
  \expandafter\newcommand\csname nomcours#2\endcsname{#3}
}

\newcommand\INFcinqsixun[3]{
  \expandafter\newcommand\csname numerocours#2\endcsname{#1}
  \expandafter\newcommand\csname nomcours#2\endcsname{#3}
}

\newcommand\MPRI[3]{
  \expandafter\newcommand\csname numerocours#2\endcsname{#1}
  \expandafter\newcommand\csname nomcours#2\endcsname{#3}
}

\newcommand\COURSEINCLUDE[2][]{
  \renewcommand\IFNOTSUB[1]{}
  \renewcommand\MYCOURSE[1]{\chapter{\ifnotmodeDIFFUSIONFINALE{#1} \csname nomcours##1\endcsname}
    \ifnotmodeDIFFUSIONFINALE{\mychapter{Fichier: #2}}
  }
   \renewcommand\MYCOURSEnew[1]{\chapter{\ifnotmodeDIFFUSIONFINALE{#1} ##1 }
    \ifnotmodeDIFFUSIONFINALE{\mychapter{Fichier: #2}}
  }
  \input{#2}
}

\newcommand\STRESSCOURSEINCLUDE[1]{
  \COURSEINCLUDE[!STRESS!]{#1}
  }

  \newcommand\subCOURSEINCLUDE[1] {
    \NDLR{BEGIN:  J'inclus/ Je répète ICI #1 }
    \input{INCLUDE/#1}
    \NDLR{END:  J'inclus/ Je répète ICI #1 }
}

\newcommand\montruc[2][]{
  \ifnotmodeDIFFUSIONFINALE{
  \begin{maboitetruc}[#1]{\danger} #2
  \end{maboitetruc}
}
}

\newcommand\montrucdanger[3][]{
  \ifnotmodeDIFFUSIONFINALE{
                    \montruc[#1]{\danger #2: #3}
                    }
}

\newcommand\montrucdangermargin[3][]{
  \ifnotmodeDIFFUSIONFINALE{
  \montruc[#1]{\danger #2: #3}
\marginpar{\danger #2: #3}
}
 }

\makeatletter
\renewcommand\beginmarkversion{\montrucdanger[colback=yellow!10]{=== begin NOT}{}\@Vs@sffbox{\@currenvir$>$}}
 \renewcommand\endmarkversion{\@Vs@sffbox{$<$\@currenvir}\montrucdanger[colback=yellow!10]{end=== NOT}{}}
\makeatother

 \newcommand\NDLR[1]{\montrucdanger[colback=blue!20]{NDLR}{#1}}
 \newcommand\NLDR[1]{\NDLR{#1}}
 \newcommand\VOIRAUSSI[1]{\montrucdanger[colback=yellow!20]{VOIR AUSSI}{#1}}
 \newcommand\TODO[1]{\montrucdanger[colback=green!20]{TODO}{#1}}
 \newcommand\ATTENTION[1]{\montruc[colback=brown!20]{\danger ATTENTION#1}}
 \newcommand\COMPRENDPAS[1]{\montrucdanger[colback=yellow!20]{JE
     COMPRENDS PAS CA}{#1}}
  \newcommand\POSSIBLE[1]{\montrucdanger[colback=yellow!10]{
      Serait possible}{#1}}
 \newcommand\ENLEVE[1]{\montrucdanger[colback=gray!20]{
      ENLEVE}{#1}}

  \newcommand\DEPLACE[1]{
    \vspace{0.5cm}
    \hrule
    \hrule
      \vspace{0.5cm}
    \montrucdangermargin[colback=black!40]{DEPLACE/MISSING CONCEPT
    }{#1}
      \vspace{0.5cm}
    \hrule
    \hrule
      \vspace{0.5cm}
  }

\newcommand\BESOINDE[1]{\montrucdangermargin[colback=pink!20]{ATTTENTION
     BESOIN DE}{#1}}
\newcommand\MISSINGCONCEPT[1]{\montrucdangermargin[colback=pink!20]{MISSING CONCEPT
     }{#1}}
\newcommand\SAUTE[1]{\montrucdangermargin[colback=pink!20]{SAUTE
  }{#1}}
\newcommand\SHOULDBEHERE[1]{\BEGINTEMPENVCOLOR{blue!30}\montrucdanger[colback=pink!10]{
SHOULD
    BE HERE
  }{    
    \vspace{1cm}
    
    \centerline{$\vdots$}
    
   \centerline{ #1}
    
      \centerline{$\vdots$}
          
    \vspace{1cm}
    }\ENDTEMPENVCOLOR}
\newcommand\COULDBEHERE[1]{\montrucdangermargin[colback=pink!20]{COULD
    BE HERE
  }{#1}}

\newcommand\SOURCEURL[1]{
  \ifnotmodeDIFFUSIONFINALE{
    \marginpar{ \textcolor{magenta}{Source: \url{#1}} }
    }
}

\newcommand\SOURCETXT[1]{
  \ifnotmodeDIFFUSIONFINALE{
    \marginpar{ \textcolor{magenta}{Source: {#1}} }
    }
}

\newcommand\SOURCE[1]{
  \nocite{#1}
  \ifnotmodeDIFFUSIONFINALE{
    \marginpar{ \textcolor{magenta}{Source: \cite{#1}} }
    }
}

\newcommand\SOURCEJECH{\SOURCE{jech2003set}}
\newcommand\SOURCEKECHRIS{\SOURCE{kechris2012classical}}
\newcommand\SOURCENOTESSETTHEORY{\SOURCE{moschovakis2006notes}}
\newcommand\SOURCEKOZENCT{\SOURCE{LivreKozenTC}}
\newcommand\SOURCEMARKER{\SOURCE{marker2002descriptive}}
\newcommand\SOURCEMOSCHO{\SOURCE{moschovakis2009descriptive}}

\ifnotmodeDIFFUSIONFINALE{
\newenvironment{recherche}
{}{}
\tcolorboxenvironment{recherche}{colback=blue!5!white, colframe=red!75!black,fonttitle=\bfseries, colbacktitle=blue!85!black,enhanced,
attach boxed title to top center={yshift=-2mm},
  title={FAIRE RECHERCHE ICI}}
}

\ifnotmodeDIFFUSIONFINALE{
\newenvironment{commentaire}{\begin{maboitecommentaire}[colback=brown!30]{Commentaire} Commentaire :}{\end{maboitecommentaire}}
\newcommand\COMMENTAIRE[2][Commentaire ]{
\begin{maboitecommentaire}[colback=brown!30]{#1}
Commentaire : #2
\end{maboitecommentaire}
}
}

\ifnotmodeDIFFUSIONFINALE{
\newenvironment{pasbesoinpourlinstant}{NOT NEEDEED FOR
  THIS SECTION: }{ }
    \tcolorboxenvironment{pasbesoinpourlinstant}{colback=green!5!white, colframe=gray!75!black,fonttitle=\bfseries, colbacktitle=green!85!black,enhanced,
attach boxed title to top left={yshift=-2mm},
  title={PAS BESOIN POUR L'INSTANT}}
}

\ifnotmodeDIFFUSIONFINALE{
\newenvironment{bonus}{BONUS: }{
  }
  \tcolorboxenvironment{bonus}{colback=gray!5!white, colframe=gray!75!black,fonttitle=\bfseries, colbacktitle=gray!85!black,enhanced,
attach boxed title to top left={yshift=-2mm},
  title={BONUS}}
}

\ifnotmodeDIFFUSIONFINALE{
\newcommand\ATTENTIONREPETE[1]{
  \ATTENTION{Repete (attention modif): #1}
  \begin{remarkolivier}
Repete (attention modif): #1
\end{remarkolivier}
}
}

\newcommand{\x}{\times}
\newcommand{\s}{\sigma}
\newcommand\sep{.}
\newcommand\bzero{\zero}
\newcommand\bun{\un}
\renewcommand{\a}{\alpha}
\def\hd{{\sf hd}}
\def\tl{{\sf tl}}
\def\cons{{\sf cons}}
\def\Cons{{\sf Cons}}
\def\Pr{{\sf Pr}}
\def\Op{{\sf Op}}
\def\Rel{{\sf Rel}}
\def\Equal{{\sf Equal}}
\def\Eval{{\sf Eval}}
\def\Gate{{\sf Gate}}
\def\Selection{{\sf Selection}}
\def\Result{{\sf Result}}
\def\Circuit{{\sf Circuit}}
\def\next{{\sf next}}
\def\comp{{\sf comp}}
\def\finalcomp{{\sf finalcomp}}
\def\Pick{{\sf Pick}}
\def\Hd{{\sf Hd}}
\def\Tl{{\sf Tl}}
\def\RevHd{{\sf RevHd}}
\def\Node{{\sf Node}}
\def\Length{{\sf Length}}
\def\Tape{{\sf Tape}}
\def\Head{{\sf Head}}
\def\Time{{\sf Time}}
\def\unpad{{\sf unpad}}
\def\concat{{\sf concat}}
\def\rg{{\sf right}}
\def\lf{{\sf left}}
\def\st{{\sf state}}
\def\tope{{\sf top}}

\newcommand\buglst[1]{#1}

\ifnotSLIDE{
\LANGUE{
\spnewtheoremi{exemple}{\NDLR{ECRIT EXEMPLE AU LIEU DE
  EXAMPLE}}{\bfseries}{\rmfamily}{theorem} 
\spnewtheoremi{remarque}{\NDLR{ECRIT REMARQUE AU LIEU DE REMARK}}{\bfseries}{\rmfamily}{theorem} 
}
{
\spnewtheoremi{exemple}{Exemple}{\bfseries}{\rmfamily}{theorem} 
\spnewtheoremi{remarque}{Remarque}{\bfseries}{\rmfamily}{theorem} 
}
}
            
            \newcommand\PERSOSOURCEPOSSIBLE[1]{\PERSOSOURCE{#1}\PERSOCOMMENTAIRE{#1}}
\newcommand\PERSOPOSSIBLE[1]{\POSSIBLE{#1}}
\newcommand\PERSOETAIT[1]{\NDLR{ETAIT: #1}}

\providecommand\nl{}

\newcommand\PERSOSOURCE[1] { \SOURCETXT{#1}}
\newcommand\PERSOCOMMENTAIRE[1]{\NDLR{#1}}
\newcommand\PERSOCOMMENTAIREd[2]{\NDLR{#1 \textbf{#2}}}

\newenvironment{dessinpb}
{
\begin{center}\begin{tikzpicture}[baseline=(current bounding
      box.north)]
}
{\end{tikzpicture}
\end{center}}

\newcommand\PERSOMANQUE[1]{\NDLR{\textbf{-- manque --
      #1 --manque--}}}

\newcommand\PERSOENLEVE[1]{\ENLEVE{#1}}

\tikzstyle{nd} = [shape=circle,draw]

\newcommand\mmod{mod}

\newcommand\mmV{v}
\newcommand\mmVb{\overline{\mmV}}
\newcommand\mmVe{\mmVb}

\newcommand\exoref[1]{\ref{exo:#1-stt}}

\newenvironment{correction}[1]
{\par\medskip\noindent \label{exo:#1-cor} \hbox{\bf Exercice \ref{exo:#1-stt} 
  \   (page \pageref{exo:#1-stt}). }}
{\par\medskip}

\newcommand\DECISIONint[3]{
\begin{itemize}
\item[\textbf{Donnée}:] #2
\item[\textbf{Réponse}:] #3
\end{itemize}
}

\newcommand\INDECIDABLE[5]{
\DECISION{#1}{#2}{#3}

\begin{#4} \label{#5}
Le problème $#1$ 
est indécidable.
\end{#4}
}

\newcommand\INDECIDABLEenglish[5]{
\DECISIONenglish{#1}{#2}{#3}

\begin{#4} \label{#5}
The problem $#1$
is undecidable.
\end{#4}
}

\newcommand\INDECIDABLEq[5]{
\DECISION{#1}{#2}{#3}

\begin{#4} \label{#5}
Le problème $#1$
est ?
\end{#4}
}

\newcommand\INDECIDABLEqenglish[5]{
\DECISION{#1}{#2}{#3}

\begin{#4} \label{#5}
The problem $#1$
is ?
\end{#4}
}

\newcommand\DECISIONintenglish[3]{
\begin{itemize}
\item[\textbf{Input}:] #2
\item[\textbf{Answer}:] #3
\end{itemize}
}

\newcommand\DECISIONi[4]{
\item $#1$
\DECISIONint{#1}{#2}{#3}
}

\newcommand\DECISIONienglish[4]{
\item $#1$
\DECISIONintenglish{#1}{#2}{#3}
}

\newcommand\DECISION[3]{
\begin{definition}[$#1$]\ 

\DECISIONint{#1}{#2}{#3}
\end{definition}
}

\newcommand\DECISIONenglish[3]{
\begin{definition}[$#1$]\ 
\DECISIONintenglish{#1}{#2}{#3}
\end{definition}
}

\newcommand\NPC[4]{ 
\DECISION{#1}{#2}{#3}

\begin{#4}
Le problème $#1$
est $\NP$-complet. \LANGUE{\myindex{NP-completeness}}{\myindex{NP-completude@$\NP$-complétude}}
\end{#4}
}

\newcommand\NPCenglish[4]{ 
\DECISIONenglish{#1}{#2}{#3}
\begin{#4}
The problem $#1$
is $\NP$-complete. \LANGUE{\myindex{NP-completeness}}{\myindex{NP-completude@$\NP$-complétude}}
\end{#4}
}

\newcommand\NPCp[5]{
\DECISION{#1}{#2}{#3}

\begin{#4}[#5]
Le problème $#1$
est $\NP$-complet.
\end{#4}
}

\newcommand\NPCpenglish[5]{
\DECISIONenglish{#1}{#2}{#3}

\begin{#4}[#5]
The problem $#1$
is $\NP$-complete.
\end{#4}
}

\newcommand\BRUNOCOMMENTAIRE[1]{\NDLR{Bruno commentaire: #1}}

\ifnotSLIDE{
\spnewtheorem{notation}{Notation}{\bfseries}{\rmfamily}
}

 \newcommand\motnouv[1]{\emph{#1}\index{#1}}
\newcommand\INTmotnouv[1]{\emph{#1}}
\newcommand\myindex[1]{\index{#1}}

\newcommand\synonyme[2]{\kindex{#1|see{#2}}}
\newcommand\indexM[1]{\index{$#1$}}
\newcommand\idx[1]{#1\index{#1}}
\newcommand\idxm[1]{#1\index{$#1$}}
\newcommand\idxM[1]{$#1$\index{$#1$}}
\newcommand\idxmm[2]{#1\index{#2@$#1$}}
\newcommand\idxi[2]{#1\index{#2}}
 \newcommand\motnouvi[2]{\INTmotnouv{#1}\index{#2}} 
 \newcommand\motnouva[2]{\INTmotnouv{#1}\index{#2@#1}} 
\newcommand\motnouvm[1]{\INTmotnouv{#1}} 

\newenvironment{PERSOVERBATIM}{}{}
 \newcommand{\joNP}[3]{{ Problème \textbf{ { \motnouv{##1}}:}
            \begin{itemize}
            \item  \textbf{ Instance:} {##2} 
            \item   \textbf{ Question:} {##3}
            \end{itemize}
          }}

\newcommand\defin[1]{\textbf{#1}}

\newcommand\soustitre[1]{\textbf{#1}}

\providecommand\motnouv[1]{\emph{#1}}
\providecommand\motnouvi[2]{\emph{#1}}

\newcommand\papiernorm[1]{\myop{norm}}

\newcommand\olivier[1]{ \ATTENTION{Olivier: #1}}

\renewcommand\round{\cp{round}}
\newcommand\dist{\cp{dist}}
\newcommand\Step{\cp{Step}}
\newcommand\Decode{\cp{Decode}}

\title{Primitive Recursion without Composition\\
{\large Dynamical Characterizations, 
from Neural Networks to Polynomial ODEs}}
\titlerunning{Primitive Recursion without Composition}
\author{Olivier Bournez}
{Institut Polytechnique de Paris, 
Ecole Polytechnique, LIX, Palaiseau, France}
{olivier.bournez.polytechnique.fr}
{}{}
\authorrunning{Olivier Bournez}
\Copyright{Olivier Bournez}

\begin{document}





\maketitle

\begin{abstract}
What computational mechanisms do recurrent neural networks,
polynomial ordinary differential equations, and discrete polynomial
maps each bring to the table, and what do they lack? All three are
models of computation over the continuum: they operate on
real-valued states and evolve by real-valued dynamics, even when the
functions we ask them to compute are ultimately discrete. We
investigate how these models compare, their strengths, their
limitations, and the precise resources on which each one relies,
through the lens of primitive recursive functions.

We prove that the classical notion of primitive recursion admits
equivalent characterizations in all three dynamical frameworks:
bounded iteration of a fixed recurrent ReLU network, robust
computation by a fixed polynomial ordinary differential equation,
and iteration of a fixed polynomial map in discrete time with an
externally supplied step-size parameter. In each case, the time
bound is itself primitive recursive, composition is not postulated
as a closure rule but emerges from the dynamics, and the input is
given as a raw integer vector with no auxiliary encoding. At the
proof level, every primitive recursive function is first compiled
into bounded iteration of a single threshold-affine normal form map,
which is then interpreted as a recurrent ReLU computation on the
one hand, and as a robust polynomial ODE on the other.

The equivalences expose a structural asymmetry between discrete and
continuous polynomial computation. We prove that no fixed
polynomial map can round uniformly toward the nearest integer, and
that none can realize exact phase selection: two operations that
polynomial ODEs perform robustly  through their continuous-time
flow. Each formalism compensates for a limitation that the others
do not share: the ReLU gate provides exact branching, continuous
time provides autonomous rounding and control, and the step-size
parameter recovers both at the cost of discretization precision.
Our equivalence theorem characterizes what each resource
contributes, and opens the way to dynamical characterizations of
subrecursive hierarchies and complexity classes by restricting the
time bounds, polynomial degrees, or discretization resources within
the same framework.

More broadly, the constructions reveal that these real-valued models
do not compute by composing subroutines in the classical sense:
they compute by shaping the trajectory of a dynamical system,
through clocks, phase selectors, stabilization mechanisms, and
error correction built into the dynamics itself. This is a mode of
computation that differs structurally from symbolic programming,
and our equivalence theorem provides a precise framework in which
the difference can be studied.
\end{abstract}

%
%


\addtocounter{page}{-1}

\newcommand\myparagraph[1]{\medskip \noindent\textbf{#1}}
\newcommand\textcite[1]{\cite{#1}}

\section{Introduction}

Over the integers, the theory of computation rests on a remarkably
stable foundation: Turing machines, register machines, the
$\lambda$-calculus, and many other formalisms are all known to define
the same notion of effective computation. Over the reals, the situation
is less settled. Several models have been proposed with different
motivations, including computable analysis \cite{Wei00} and algebraic
models \cite{BSS89}, but they are provably inequivalent
\cite{Wei00,bournez2021survey}.
Within this landscape, two families of models have recently emerged as
\emph{genuinely computational} rather than merely descriptive, and both
have developed substantial theory.

The first comes from neural computation. Already in the early 1990s,
recurrent networks with threshold-like or analog activations were shown
to possess significant computational power and nontrivial
complexity-theoretic structure \cite{SS94,BGSS93,Maa94a,LivreSiegelmann}.
More recent architectures, most notably transformer models, rely on
substantially different mechanisms, and only make these foundational
questions more pressing. What matters here is not just that such models
are expressive, but that they are \emph{programmable}, in a way quite
unlike classical symbolic programming.

The second family comes from the analog-computing (General Purpose Analog Computer \cite{Sha41}) tradition and, more broadly, from
polynomial Ordinary Differential Equations (ODEs). A long line of work has
established polynomial ODEs as a robust and machine-independent
framework for computability and complexity over the reals
\cite{CM06b,BH05,BouCamGra06a,CIE2013Daniel,ICALP2016ventardise,BGDPRiccardo2022,BlancBournezMFCS2023Journal}.
Again, the point is not merely that such systems can simulate classical
models, but that computations can be synthesized through mechanisms
natural to continuous dynamics: oscillation, stabilization, transfer of
intermediate quantities, and correction of errors along the flow.

Despite their apparent differences, these two families share a common feature:
they realize computation by shaping the evolution of a dynamical system. This
is a very different constructive logic from the classical one based on
explicit symbolic composition. In recurrent networks and in polynomial ODEs,
one often obtains a computation not by writing down its decomposition into
subroutines, but by designing a state space, a feedback mechanism, a
thresholding behavior, or a flow with the right dynamical properties.

A first indication that this viewpoint is genuinely different is
that it changes what is easy and what is hard. Rounding to the
nearest integer, for instance, is almost free in continuous time
(a polynomial ODE built on $y''=-y$ contracts to the lattice) but
is provably costly or impossible for a pure polynomial map in
discrete time. Section~\ref{sec:discrete-vs-continuous} makes this
asymmetry precise; it is the dynamical fact that organizes the
rest of the paper.


This leads to a natural foundational question. Instead of taking machines or
syntactic closure principles as the primary point of departure, can one find a
simple classical class of functions that already sits at the intersection of
these dynamical modes of computation? Through this paper, we argue that the right
place to begin is to get back to the historical roots of computable theory: the old but remarkably robust class of primitive recursive
functions. We assume the reader has some familiarity with computability theory, particularly with primitive recursive functions: \cite{Odi99} provides a thorough reference.

\myparagraph{Primitive recursion as a dynamical system}
The class~$\mathcal{PR}$ of primitive recursive functions is usually
defined as the smallest class of functions over the natural numbers
containing the zero, successor, and projection functions, and closed
under composition and primitive recursion. Many subrecursive and
complexity-theoretic classes are obtained by restricting these
schemes, including the Grzegorczyk hierarchy
\cite{Grz55,Grz53a}, polynomial-time functions \cite{Cob65}, and
polynomial-space functions \cite{thompson1972subrecursiveness}; see
\cite{clote2013boolean} for a broad overview. In all these
definitions, closure under composition is taken for granted: this is somehow  the
one rule that no formalism considers optional.

Primitive recursion already admits a dynamical reading. The scheme
$$
f(0,\vx)=g(\vx), \qquad f(n+1,\vx)=h(n,f(n,\vx),\vx)
$$
can be viewed as the bounded iteration of the transition
$(n,a,\vx)\mapsto(n+1,h(n,a,\vx),\vx)$, where $a$ plays the role of
an accumulator: in this view, the computation is not assembled from
composed subroutines, but unfolds as a single trajectory. This is
reminiscent of how polynomial ODEs compute in the analog tradition,
where a single continuous flow replaces any composition rule, and the
required control emerges from the dynamics itself through basins of
attraction, stabilization, and autonomous switching mechanisms.

The same pattern appears in recurrent neural networks, though this
is not always made explicit. A recurrent network with
activation~$\rho$ repeatedly applies a fixed update rule to a state
vector, where one step is a finite composition of affine maps and
coordinatewise applications of~$\rho$. The computation is the
trajectory $z_{n+1}=R(z_n)$ of a fixed dynamical system. When
$\rho$ is $\operatorname{ReLU}(t)=\max(t,0)$, the system is
piecewise linear. When $\rho$ has bounded range,  as for the
sigmoid activations used in practice,  the dynamics is confined to
a bounded region of state space.

 This leads to a concrete version
of our foundational question: can primitive recursive functions be
characterized as the functions computable by bounded iteration of
such a recurrent network, with no explicit closure under composition as a basic external assumption, but composition emerging from the dynamics.

\myparagraph{Main results.} 
The answer is yes, and the characterization extends to all the
dynamical frameworks discussed above. Our main theorem states:

\begin{theorem}[Main theorem, informal\footnote{Formal version is Theorem \ref{thm:main-formal}}]
\label{th:unf}
For a function $f:\N^d\to\N^e$, the following are equivalent:
\begin{enumerate}
\item \label{iteminformalpr} $f$ is primitive recursive in the classical sense;
\item \label{iteminformalrelu} $f$ is computable by bounded iteration of a fixed recurrent
  ReLU network;
\item \label{iteminformalrho} $f$ is computable by bounded iteration of a fixed recurrent
  $\rho$-activation network with 
  bounded precision, over a compact domain;
\item \label{iteminformalpode} $f$ is computable robustly by a fixed polynomial ODE;
\item \label{iteminformaldodenew}  $f$ is computable by iteration of a fixed polynomial map in
  discrete time, provided a sufficiently small step size is supplied
  as an external parameter.
\end{enumerate}
\end{theorem}

In each case, ``bounded'' means that the relevant quantity, e.g. the
number of iterations, the observation time, the domain diameter, or
the precision, is itself primitive recursive as a function of the
input.   The activation~$\rho$ is only required to be primitive
recursively computable in a natural sense, a condition that is satisfied by all sigmoid functions considered in practice.

In Items~(\ref{iteminformalrelu}) and~(\ref{iteminformalpode})--(\ref{iteminformaldodenew}), a crucial point is that the computation
starts from the raw input state $(x,0,\ldots,0)$: the integer
input~$x$ is embedded directly in~$\R^m$, with no auxiliary
encoding.  Item~(\ref{iteminformalrho}) necessarily departs from this convention, since
a bounded domain cannot accommodate arbitrarily large integers as
raw values; there, the input is supplied as a rational number under
a fixed, standard encoding,  but, in the same spirit, no
problem-specific preprocessing is involved. Namely, we consider $\nu(n)=2^{-n}$ to encode integer $n$, i.e we encode $n$ as a dyadic, but this choice is  arbitrary. 

\begin{remark}
Our raw-input assumption is essential, to reveal the strength of the statement and, not merely a stylistic
choice.  Without it, the results would be considerably easier to
prove but would say far less: the computational work could be
hidden in the encoding step rather than exposed in the dynamics.
As stated, our theorems guarantee that the entire
computation  is carried out by the dynamics of the
system itself, with no external help.
\end{remark}

The equivalence $(\ref{iteminformalpr})\Leftrightarrow(\ref{iteminformalrelu})$ says that composition is
redundant: bounded iteration of a single ReLU network already
generates the full class.  The equivalence
$(\ref{iteminformalrelu})\Leftrightarrow(\ref{iteminformalrho})$ shows that any well-behaved activation $\rho$ yields the same
computational power, provided domain and precision are
controlled.  The passage $(\ref{iteminformalrelu})\Rightarrow(\ref{iteminformalpode})$ eliminates the
threshold gates by moving to continuous
time, where a polynomial ODE autonomously generates the required
control through
mechanisms like the one described above.  The passage
$(\ref{iteminformalpode})\Rightarrow(\ref{iteminformaldodenew})$ discretizes the ODE, but the control that was
free in continuous time now reappears as the step-size parameter: a
genuine computational resource.  The converse
$(\ref{iteminformaldodenew})\Rightarrow(\ref{iteminformalrelu})$ states that once the 
map and the time
bound are fixed, the induced evolution on rational states is
primitive recursive.

These equivalences make visible genuinely different \emph{costs of
computation}: the ReLU model treats threshold switching as free,
the $\rho$-model pays with domain and precision, the polynomial
ODE pays with continuous time, and the discrete polynomial model
pays with step size.

{
\myparagraph{What these characterizations reveal}
The five formulations are equivalent, but they make visible
genuinely different \emph{costs of computation}. In the ReLU
model~(\ref{iteminformalrelu}), the implicit threshold gate provides exact switching for
free, but this is a strong assumption on the activation. The
$\rho$-activation model~(\ref{iteminformalrho}) shows that any reasonable nonlinearity
suffices, at the price of controlling domain and
precision. In the polynomial ODE~(\ref{iteminformalpode}), no discrete activation is
needed at all: the dynamics itself generates rounding, switching,
and error correction, but requires continuous time. In
the discrete polynomial simulator~(\ref{iteminformaldodenew}), continuous time is also
removed, and the step size becomes the price paid for its
absence.}

\MFCSSHORTER{
\begin{remark}
The continuous-time ODE model may seem remote from neural networks,
but the connection is direct: deep residual networks provably
approach continuous-time ODEs as the number of layers
grows: this is the well-known neural ODE approach~\cite{chen2018neural,kidger2022neural,ruthotto2024differential},
and the mechanisms that make polynomial ODE computation feasible
(autonomous rounding, phase-controlled switching, error correction)
are, we believe, the same mechanisms that underlie the computational
power of deep recurrent networks in discretized form.  That said,
the five models do not succeed for the same reasons: each relies on
structural features that the others lack, and the equivalence holds
only under carefully matched hypotheses on the resources (activation
type, domain boundedness, precision, step size) that compensate for
these differences.  The situation is analogous to the Church--Turing
thesis~\cite{Gur99}: Turing machines, $\lambda$-calculus, and
register machines compute the same functions, but this extensional
equivalence says nothing about the different methods by which each
model organizes its computations.  Making the compensation between
our five models precise, rather than papering over it, is the point
of the theorem.
\end{remark}}

\MFCSSHORTER{
\myparagraph{Perspective}
The characterizations suggest that primitive
recursion is the first level of a broader theory of computation in
which the fundamental objects are dynamical systems and the
fundamental resources are structural parameters of the dynamics.
Since each resource can be independently restricted, it is natural
to ask whether doing so recovers finer complexity-theoretic classes
(the Grzegorczyk hierarchy, $\mathsf{FP}$,
$\mathsf{FPSPACE}$) within the same framework.  We
leave this as a concrete direction for future work.
The techniques we develop (polynomial clocks, phase selectors,
scrubbing mechanisms) are the same building blocks that appear in
recent characterizations of $\mathsf{P}$ and $\mathsf{PSPACE}$ by
polynomial ODEs.  The present work shows that these techniques are
already meaningful at the most basic recursion-theoretic level.
}

\myparagraph{Related work}
\label{sec:related-work}
Our work sits at the intersection of three traditions. Classical
\emph{subrecursive function theory} characterizes complexity classes
by restricting recursion schemes — Cobham for $\mathsf{P}$
\cite{Cob65}, Bellantoni--Cook~\cite{BC92}, Leivant--Marion for
$\mathsf{FPSPACE}$~\cite{LM97}, see~\cite{clote2013boolean} for a
survey — but always takes composition as a built-in closure. The
\emph{neural-computation} line~\cite{SS95,SS94,BGSS93,Maa94a}
showed early on that recurrent analog systems compute through state
update, feedback, and thresholding rather than symbolic
composition. The \emph{GPAC and polynomial-ODE} tradition, from the
Grzegorczyk-level analog characterizations
of~\textcite{Cam01-f,CamMooCos00a,BH05,GraCos03} to the
polynomial-time and polynomial-space characterizations
of~\cite{JournalACM2017,TheseAmaury,TheseRiccardo,GozziPSPACE} and
the discrete-ODE lines
of~\cite{MFCS2019,MFCSJournal,Antonelli0K24,Antonelli0K25,AntonelliDK26},
shows that polynomial dynamics provide a machine-independent
framework over the reals. Our paper combines these: machine-free
like implicit complexity, polynomial dynamics like the ODE
literature, discrete-time and raw-input like neither, and the
internalization of composition as a consequence of the dynamics
rather than a postulate. 

\section{The asymmetry between discrete and continuous polynomial dynamics}
\label{sec:discrete-vs-continuous}

\myparagraph{Computation by a polynomial dynamical system}
Before stating the results of this section, we make precise the common
shape of computation shared by items~(\ref{iteminformalpode})
and~(\ref{iteminformaldodenew}) of the main theorem, and by the
dynamical reading of items~(\ref{iteminformalrelu})
and~(\ref{iteminformalrho}) as well. The data is a polynomial map
$p:\R^m \to \R^m$, viewed as a fixed vector field on the state space
$\R^m$. This single object can be evolved in two ways. In discrete
time, it defines the iteration $x_{n+1} = p(x_n)$. In continuous time,
it defines the autonomous ODE $\dot{x} = p(x)$. In both cases a
computation on input $x \in \N^d$ is a trajectory started from the raw
initial state $(x, 0, \ldots, 0)$, and the output is read off a
designated coordinate at some time $n$ (respectively $t$) that depends
on the input.

What is striking, and what this section is about, is that these two
evolutions of the same polynomial data are not computationally
interchangeable. Continuous time makes available mechanisms, such as
contraction toward the integers and autonomous phase switching, that a
discrete iteration of the same kind of object provably cannot realize
uniformly. Throughout this section, $f^{[n]}:=f\circ\cdots\circ f$ ($n$~times)
denotes the $n$-fold iterate of a map~$f$, with $f^{[0]}:=\mathrm{id}$.

\myparagraph{Two primitives of dynamical computation}
To make the asymmetry quantitative, we examine two operations that are
natural building blocks for any computation on integer inputs carried
out by a dynamical system. The first is \emph{rounding}: contracting a
real state that is close to an integer back onto that integer, so that
integer values can be maintained along a long trajectory. The second
is \emph{phase selection}: switching autonomously between phases of
the computation after a number of steps supplied as input. Both are
ubiquitous in the GPAC and polynomial-ODE literature, where they serve
as the foundation on which clocks, error correction, and bounded
recursion are built~\cite{dsg05}.

We now state two impossibility results showing that both of these
operations, which are essentially free in continuous time for
polynomial vector fields, are provably unattainable by pure polynomial
maps in discrete time. The proofs are short and direct, but the
consequence is conceptual: the discrete and continuous settings rest
on genuinely different assumptions about what is elementary, and this
explains why additional computational resources (the
$\rho$-activation gate in the neural models, the step-size parameter
in the discrete simulator) become mandatory when passing from one
world to the other.

\begin{remark}
We focus here on the contrast between discrete-time and continuous-time
polynomials, but the conclusions bear directly on the neural-network
setting. The choice of activation function is itself a commitment about
which operations are taken as primitive. Our impossibility results make
this commitment explicit: the activation~$\rho$ is not a convenience
but a necessary compensation for what polynomials alone cannot do in
discrete time. Since affine maps are particular polynomial maps, any
impossibility for polynomial maps \emph{a fortiori} rules out affine
maps, and thus pure feedforward linear networks, as well.
\end{remark}

\myparagraph{Uniform rounding}
One  basic operation when passing between integers and reals is
rounding. As observed in \cite{dsg05}, in continuous time, this has an elegant and essentially costless
solution: The map
$$
\sigma(x) = x - \frac{1}{2\pi}\sin(2\pi x)
$$
fixes every integer and contracts every real number toward the nearest
integer: for any $\varepsilon < 1/2$, there exists $\lambda < 1$ such
that $|\sigma(n+\delta) - n| \leq \lambda\,|\delta|$ whenever
$|\delta| \leq \varepsilon$. Since $\sin$ is natively generated in the
polynomial ODE world by $y'' = -y$, this contraction is available at no
additional cost. This mechanism, introduced as the very first statement (Lemma 1) of \cite{dsg05},
is the foundational building block of the entire GPAC computation theory:
all subsequent control structures, such as clocks, holding phases, error
correction, are built on top of it in \cite{dsg05}. 

In discrete time, the analogous task is provably impossible from a fixed
polynomial map:

\begin{propositionrep}[No uniform polynomial rounder in discrete time]
\label{prop:no-discrete-rounder}
Let $0 < \varepsilon < \delta < 1/2$. There is no polynomial
$p:\R \to \R$ and no integer $N \geq 1$ such that for every $m \in \Z$
and every $u$ with $|u - m| \leq \delta$,
$
|p^{[ N]}(u) - m| \leq \varepsilon.
$
\end{propositionrep}

\begin{proof}
Suppose such $p$ and $N$ exist. The polynomial $q(x) = p^{[ N]}(x) - x$
satisfies $|q(m)| \leq \varepsilon$ for every integer $m$. A polynomial
bounded on infinitely many integers is constant: $q \equiv c$. But then
$|{\delta + c}| \leq \varepsilon$ and $|{-\delta + c}| \leq \varepsilon$,
which is impossible when $\varepsilon < \delta$.
\end{proof}

To achieve approximate rounding in discrete time, one can approximate
$\sigma$ on a bounded interval $[-B, B]$ by a polynomial whose degree
depends on~$B$ and on the target accuracy. An operation that is free in
continuous time becomes structurally expensive in discrete time: the
polynomial degree is the currency in which the discrete construction pays.

\myparagraph{Exact phase selection}
A second fundamental operation is switching between phases of a
computation. Given a counter $C$ supplied as input, one needs an
indicator that outputs~$1$ for the first $C$ steps and~$0$ thereafter.
In continuous time, a polynomial ODE comparator achieves this
autonomously, driving a register exponentially fast toward~$0$ or~$1$
depending on the sign of its input. This can be done using a polynomial map \cite{dsg05}. 
\begin{propositionrep}[No exact polynomial phase selector]
\label{prop:no-exact-selector}
There is no polynomial map $P:\R^m \to \R^m$ and coordinate
projection $\pi:\R^m \to \R$, $(z_1,\ldots,z_m)\mapsto z_i$, such that
for every integer $C \in \N$,
$$
\pi(P^{[n]}(C,0,\ldots,0)) =
\begin{cases}
1 & \text{if } n < C,\\
0 & \text{if } n \geq C,
\end{cases}
\qquad \text{for all } n \in \N.
$$
\end{propositionrep}

\begin{proof}
For fixed~$n$, the function $Q_n(C) = \pi(P^n(C,0,\ldots,0))$ is a
polynomial in~$C$. Since $Q_n(C) = 1$ for all integers $C > n$, the
polynomial $Q_n - 1$ has infinitely many roots, so $Q_n \equiv 1$.
But $Q_n(n) = 0$ by hypothesis---a contradiction.
\end{proof}

\section{Formal framework and equivalence for hard models}
\label{sec:formal}

We now define the five dynamical models, state the equivalence
theorems, and sketch the proof architecture.

\myparagraph{Notation}
Throughout the paper, an \emph{affine map} from $\R^p$ to $\R^q$ is a
function $A:\R^p\to\R^q$ of the form $A(u)=Mu+b$ with
$M\in\mathbb{Q}^{q\times p}$ and $b\in\mathbb{Q}^q$. Unless stated
otherwise, all affine maps considered in this paper have rational
coefficients. Given an ordered tuple of distinct indices
$I=(i_1,\ldots,i_e)\in\{1,\ldots,m\}^e$, we write
$\pi_I:\R^m\to\R^e$ for the associated coordinate projection
$\pi_I(u_1,\ldots,u_m)=(u_{i_1},\ldots,u_{i_e})$. Where the index
tuple is clear from context, we abbreviate $\pi_I$ by $\pi_e$.

\subsection{The discrete exact model: recurrent ReLU computation}
\label{sec:relu-model}

\begin{definition}[Feedforward ReLU block]
\label{def:relu-block}
A \emph{feedforward ReLU block} on $\R^m$ is a map
$R:\R^m\to\R^m$ of the form
$
R \;=\; A_L\circ\sigma\circ A_{L-1}\circ\sigma\circ\cdots
        \circ A_1\circ\sigma\circ A_0,
$
where $L\ge 0$ is a fixed integer (the \emph{depth}); the
$A_\ell$ are affine maps with rational coefficients
$A_\ell:\R^{m_\ell}\to\R^{m_{\ell+1}}$, with $m_0=m_{L+1}=m$ and
arbitrary intermediate widths $m_1,\ldots,m_L\in\N$; and $\sigma$
denotes the coordinatewise application of
$\operatorname{ReLU}(t):=\max(t,0)$.
\end{definition}

Equivalently, $R$ is the function computed by a feedforward neural
network with $L$ hidden layers, rational weights and biases, and ReLU
activations. The map $R$ is piecewise affine with rational pieces.

\begin{definition}[Recurrent ReLU computation]
\label{def:relu-unfolding}
A function $f:\N^d\to\N^e$ admits a \emph{recurrent ReLU computation}
if there exist an integer $m\ge d$, a feedforward ReLU block
$R:\R^m\to\R^m$, a primitive recursive observation time
$T:\N^d\to\N$, and an ordered tuple of designated output-coordinate
indices $I\in\{1,\ldots,m\}^e$ such that, for every input
$x\in\N^d$, the orbit
$
z_0(x):=(x,0,\ldots,0)\in\R^m$,
$z_{n+1}(x):=R(z_n(x))
$
satisfies $f(x)=\pi_I\bigl(z_{T(x)}(x)\bigr)$.
\end{definition}

In other words: a single fixed piecewise-affine map with rational coefficients is
iterated from the raw input $(x,0,\ldots,0)$, and the exact value of
$f(x)$ is read off by projecting onto the output coordinates after
$T(x)$ steps.

\subsection{The neural network model: recurrent $\rho$-computation}
\label{sec:rho-model}

For bounded-range activations (sigmoid, $\tanh$, \ldots), the compact
state space cannot accommodate arbitrarily large integers as raw
values. We fix the encoding $\nu:\N\to(0,1]$ defined by
$\nu(n):=2^{-n}$, and write
$\nu(\vx):=(2^{-x_1},\ldots,2^{-x_d})$ for a tuple $\vx\in\N^d$.
Note that $\nu(0)=1$ and $\nu(n)\le\tfrac{1}{2}$ for $n\ge 1$.

\begin{definition}[Feedforward $\rho$-block]
\label{def:rho-block}
Let $\rho:[0,1]\to[0,1]$. A \emph{feedforward $\rho$-block} of
width $m$ is a map $R:[0,1]^m\to[0,1]^m$ of the form
$
R \;=\; A_L\circ\rho\circ A_{L-1}\circ\rho\circ\cdots
        \circ A_1\circ\rho\circ A_0,
$
where $L\ge 0$ is a fixed integer depth; the $A_\ell$ are affine maps
with rational coefficients
$A_\ell:\R^{m_\ell}\to\R^{m_{\ell+1}}$, with $m_0=m_{L+1}=m$ and
arbitrary intermediate widths, satisfying
$A_\ell\bigl([0,1]^{m_\ell}\bigr)\subseteq[0,1]^{m_{\ell+1}}$ for
every~$\ell$; and $\rho$ denotes the coordinatewise application of
$\rho$ to a vector of $[0,1]^{m_\ell}$.
\end{definition}

The inclusion condition on the $A_\ell$ simply ensures that the
composition is well-defined on $[0,1]^m$ and takes values in
$[0,1]^m$.

\begin{definition}[Admissible activation]
\label{def:nu-admissible-rho}
A function $\rho:[0,1]\to[0,1]$ is \emph{admissible} if:
\begin{enumerate}[(i)]
\item $\rho$ is computable in the sense of computable
analysis~\cite{Wei00,Ko91}, with a primitive recursive evaluation
oracle and a primitive recursive modulus of uniform continuity;
\item there exist a feedforward $\rho$-block 
 $Z:[0,1]\to[0,1]$ and
some $\eta\in(0,\tfrac{1}{8})$ such that
$|Z(u)|\le\eta$ for $u\in[0,\tfrac{5}{8}]$ and
$|Z(u)-1|\le\eta$ for $u\in[\tfrac{7}{8},1]$.
\end{enumerate}
\end{definition}

Condition~(i) ensures primitive recursive simulability.
Condition~(ii) provides a robust zero-detection mechanism: the network
can distinguish $\nu(0)=1$ from all other codes. Both conditions hold
for all standard sigmoid activations.

\begin{definition}[Recurrent $\rho$-computation]
\label{def:rho-unfolding}
Let $\rho$ be admissible. A function $f:\N^d\to\N^e$ admits a
\emph{recurrent $\rho$-computation} if there exist an integer
$m\ge d$, a feedforward $\rho$-block $R:[0,1]^m\to[0,1]^m$, a
primitive recursive observation time $T:\N^d\to\N$, a primitive
recursive output precision $S:\N^d\to\N$ with $S(x)\ge 2$, and an
ordered tuple $I=(i_1,\ldots,i_e)\in\{1,\ldots,m\}^e$ of designated
output-coordinate indices, such that the orbit
$
z_0(x):=(\nu(x),0,\ldots,0)\in[0,1]^m,
\qquad z_{n+1}(x):=R(z_n(x))
$
satisfies, for every output index $j\in\{1,\ldots,e\}$ and every
$x\in\N^d$:
$
\bigl|\,\bigl(\pi_I(z_{T(x)}(x))\bigr)_j
       - \nu\bigl(f_j(x)\bigr)\,\bigr|
\;\le\; 2^{-S(x)}\,\nu\bigl(f_j(x)\bigr).
$
\end{definition}

The bound $S\ge 2$ guarantees unique decodability: for every $n\in\N$,
the nearest $\nu$-code to $\nu(n)$ is at distance
$\tfrac{1}{2}\nu(n)$, so a relative error of $2^{-S}\le\tfrac{1}{4}$
suffices. The precision function $S(x)$ plays the role of a readout
resource: the network and initial state are fixed once $x$ is given,
and only the required output precision varies with the input.

\subsection{The continuous model: polynomial ODE computation}

Here and below, a \emph{polynomial map} (resp.~\emph{polynomial vector
field}) from $\R^p$ to $\R^q$ is a function whose $q$ components are
polynomials in the $p$ input variables with rational coefficients.
Thus every affine map is in particular a polynomial map.

\begin{definition}[Robust polynomial ODE representation]
\label{def:ode-rep}
A function $f:\N^d\to\N^e$ admits a \emph{robust polynomial ODE
representation} if there exist an integer $M\ge d$, a polynomial
vector field $F:\R^M\to\R^M$ with rational coefficients, a primitive
recursive observation time $\tau:\N^d\to\N$, a primitive recursive
safety bound $\widetilde{B}:\N^d\to\N$, and an ordered tuple
$I=(i_1,\ldots,i_e)\in\{1,\ldots,M\}^e$ of designated
output-coordinate indices, such that for every $x\in\N^d$ the
solution $Y=Y_x$ of the Cauchy problem
$$
Y'(t)=F\bigl(Y(t)\bigr),
\qquad
Y(0)=(x,0,\ldots,0)\in\R^M,
$$
satisfies
(i)~$Y$ is defined on $[0,\tau(x)+1]$;
(ii)~$\|Y(t)\|_\infty\le\widetilde{B}(x)$ on this interval;
(iii)~$\bigl\|\pi_I(Y(t))-f(x)\bigr\|_\infty<\tfrac{1}{2}$ for all
$t\in[\tau(x),\tau(x)+1]$.
\end{definition}

The vector field $F$ is fixed; only the observation time, the safety
bound, and the choice of output indices depend on the problem. The
output is recovered by coordinatewise rounding of $\pi_I(Y(\tau(x)))$
to the nearest integer: condition~(iii) guarantees that every output
coordinate stays within distance $\tfrac{1}{2}$ of the integer
$f_j(x)$ throughout the entire window $[\tau(x),\tau(x)+1]$, so the
reading is insensitive to the exact observation time. This is the
sense in which the representation is \emph{robust}.

\subsection{The intermediate model: step-size-controlled polynomial simulation}

\begin{definition}[Step-size-controlled polynomial representation]
\label{def:step-controlled}
A function $f:\N^d\to\N^e$ admits a \emph{step-size-controlled
polynomial representation} if there exist an integer $M\ge d+1$, a
polynomial map $P:\R^M\to\R^M$ with rational coefficients, a
primitive recursive precision threshold $S:\N^d\to\N$, a primitive
recursive observation count $N:\N^d\times\N\to\N$, and an ordered
tuple $I=(i_1,\ldots,i_e)\in\{1,\ldots,M\}^e$ of designated
output-coordinate indices, such that for every $x\in\N^d$ and every
precision parameter $s\in\N$ with $s\ge S(x)$, the orbit defined by
$
Z_0 := (x,\,2^{-s},\,0,\ldots,0)\in\R^M,
\qquad
Z_{n+1} := P(Z_n)
$
satisfies
$\bigl\|\pi_I(Z_{N(x,s)})-f(x)\bigr\|_\infty<\tfrac{1}{2}$.
\end{definition}

We call models where the orbit remains integer-valued \emph{hard} (the ReLU model and the threshold-affine normal form below) and models where it generically leaves the integers \emph{smooth} (the $\rho$-model, the polynomial ODE, and the step-size-controlled simulator).

The map $P$ is purely polynomial: no threshold gate or transcendental
activation is involved. The step size $h=2^{-s}$ is supplied as part
of the initial state, stored in a dedicated coordinate.

%

We isolate the proof-level normal form through which all
equivalences pass.  Fix $\Theta:\R\to\R$ with $\Theta(u)=0$ for
$u\le 0$ and $\Theta(u)=1$ for $u\ge 1$; its values on $(0,1)$
are irrelevant here, since all gate arguments will be integers.

\begin{definition}[Threshold-affine normal form]
\label{def:threshold-affinenormalform}
A function $f:\N^d\to\N^e$ admits a \emph{threshold-affine normal
form} if there exist an integer~$m$, affine maps
$A_0,\dots,A_s:\R^m\to\R^m$ and $q_1,\dots,q_s:\R^m\to\R$
\emph{with integer coefficients}, a primitive recursive time bound
$T:\N^d\to\N$, and output coordinates $\pi_e:\R^m\to\R^e$ such
that, writing
\begin{equation}\label{hardnormalform}
  P(y)\;=\;A_0(y)\;+\;\sum_{j=1}^s\Theta(q_j(y))\,A_j(y),
\end{equation}
and defining $y_0(x):=(x,0,\ldots,0)$,\;
$y_{n+1}(x):=P(y_n(x))$, one has
$f(x)=\pi_e(y_{T(x)}(x))$ for every $x\in\N^d$.
\end{definition}

Since the $A_j$ and $q_j$ have integer coefficients and
$y_0(x)\in\Z^m$, every iterate $y_n(x)$ lies in $\Z^m$ and every
gate argument $q_j(y_n(x))$ is an integer; in particular, $\Theta$
is never evaluated in its ambiguous region $(0,1)$.

\begin{theorem}[Main equivalence: hard models]
\label{thm:main-formal-hard}
For a function $f:\N^d\to\N^e$, the following are equivalent:
\begin{enumerate}
\item \label{itemprhard} $f$ is primitive recursive;
\item \label{itemreluhard} $f$ admits a recurrent ReLU computation;
\item \label{itemcorhard} $f$ admits a threshold-affine normal form.
\end{enumerate}
\end{theorem}

The proof is based on the fact that  primitive
recursive functions are known to coincide with those computable by a LOOP
program~\cite{MeyerRitchie67}, and that a single LOOP instruction
compiles directly into the form given by Equation~\eqref{hardnormalform}. The interest
of the theorem lies not in the proof but in what it singles out: the
threshold-affine normal form of Item~\ref{itemcorhard}. This is the
central object of the paper, and the hinge through which the smooth
models will also be analyzed.

\myparagraph{More details on proof architecture.}
\emph{$(\ref{itemprhard})\Rightarrow(\ref{itemcorhard})$}
(Section~\ref{sec:pr-to-threshold}). A LOOP program computing $f$ is
compiled one instruction at a time into~\eqref{hardnormalform}:
arithmetic is carried by the affine maps $A_j$, and zero-testing by
the threshold $\Theta$ applied to an integer argument~$q_j(y)$.
\emph{$(\ref{itemcorhard})\Rightarrow(\ref{itemreluhard})$}
(Section~\ref{sec:threshold-to-relu}). Each threshold gate
$\Theta(q_j(y))$ is replaced by
$\operatorname{ReLU}(q_j(y))-\operatorname{ReLU}(q_j(y)-1)$, an
expression that coincides with $\Theta$ on every integer argument.
\emph{$(\ref{itemreluhard})\Rightarrow(\ref{itemprhard})$}
(Section~\ref{sec:relutopr}). A fixed ReLU block with rational
parameters induces, on rational initial states, a trajectory whose
coordinates are primitive recursive in the iteration count.

\myparagraph{A subtlety.}
The compilation
$(\ref{itemcorhard})\Rightarrow(\ref{itemreluhard})$ produces a ReLU
block that agrees with the underlying threshold-affine map only along
the integer-valued trajectories relevant to the computation, not on
all real inputs. This is enough here, since the equivalence concerns
functions on~$\N^d$.

\section{Equivalence for smooth models}

\subsection{A dynamical systems perspective}
\label{subsec:hard-to-smooth}

The equivalences for hard models in
Theorem~\ref{thm:main-formal-hard} establish more than the bare
statement. The constructions endow the threshold-affine map
$P$ of~\eqref{hardnormalform}, together with its orbits, with four
structural properties that are most naturally phrased in the
language of dynamical systems.

\emph{(i)~The integer lattice is forward-invariant under~$P$.}\;
The affine maps $A_j$ have integer coefficients, and $\Theta$ takes
values in $\{0,1\}$ on integer arguments; hence
$P(\Z^m)\subseteq\Z^m$. Every trajectory issued from the raw
integer initial condition $y_0(x)=(x,0,\ldots,0)$ therefore stays
on the lattice.

\emph{(ii)~The reference trajectory is bounded, uniformly in the
input.}\;
There exists a primitive recursive function $B:\N^d\to\N$ such
that, for every $x\in\N^d$ and every $n\le T(x)$, the state
$y_n(x)$ lies in the cube $[-B(x),B(x)]^m$. Combined with~(i),
this confines the effective dynamics to a finite subset of the
lattice $\Z^m$, whose size is primitive recursively controlled by
the input.

\emph{(iii)~A distinguished orbit carries the computational
content.}\;
For each input~$x$, the finite sequence
$y_0(x),\ldots,y_{T(x)}(x)$ is a specific trajectory of~$P$ along
which the output is read off at time~$T(x)$. The computation is
encoded by the family of these \emph{reference orbits}, indexed
by $x\in\N^d$.

\emph{(iv)~The reference orbit is robust to perturbations of the
threshold.}\;
Every gate argument $q_j(y_n(x))$ on the reference orbit is an
integer; it lies in the region where $\Theta$ is fully determined
($u\le 0$ or $u\ge 1$) and never enters the ambiguous interval
$(0,1)$. The orbit is therefore insensitive to any modification
of~$\Theta$ that preserves its values on~$\Z$.

Property~(iv) is the hinge between the hard and smooth worlds: any
alternative dynamics agreeing with~$P$ on~$\Z^m$ reproduces the
computation faithfully. Formally, we have:

\begin{theoremrep}[Uniform-robust characterization]
\label{thm:uniform-robust}
A function $f:\N^d\to\N^e$ is primitive recursive if and only if it
admits a threshold-affine normal form in which $\Theta(u)=0$ for
$u\le\tfrac{1}{4}$ and $\Theta(u)=1$ for $u\ge\tfrac{3}{4}$.
\end{theoremrep}

\begin{proof}
($\Leftarrow$)~Suppose $f$ admits a threshold-affine normal form in
which the threshold function~$\widetilde{\Theta}$ satisfies
$\widetilde{\Theta}(u)=0$ for $u\le\tfrac{1}{4}$ and
$\widetilde{\Theta}(u)=1$ for $u\ge\tfrac{3}{4}$. Since
$\{u:u\le 0\}\subseteq\{u:u\le\tfrac14\}$ and
$\{u:u\ge 1\}\subseteq\{u:u\ge\tfrac34\}$, $\widetilde{\Theta}$ in
particular satisfies $\widetilde{\Theta}(u)=0$ for $u\le 0$ and
$\widetilde{\Theta}(u)=1$ for $u\ge 1$. Hence $\widetilde{\Theta}$
is a valid instance of the threshold function fixed in
Definition~\ref{def:threshold-affinenormalform}, and $f$ admits a
threshold-affine normal form in the sense of that definition. By
Theorem~\ref{thm:main-formal-hard},
$f\in\mathcal{PR}$.

($\Rightarrow$)~Suppose $f\in\mathcal{PR}$. By
Theorem~\ref{thm:main-formal-hard}, there exist an integer $m$,
affine maps $A_0,\dots,A_s:\R^m\to\R^m$, affine maps
$q_1,\dots,q_s:\R^m\to\R$, a primitive recursive time bound
$T:\N^d\to\N$, an index tuple
$I\in\{1,\dots,m\}^e$ of output coordinates, and a threshold
function $\Theta:\R\to\R$ satisfying $\Theta(u)=0$ for $u\le 0$ and
$\Theta(u)=1$ for $u\ge 1$, such that, writing
\begin{equation*}
P(y)\;=\;A_0(y)+\sum_{j=1}^{s}\Theta\!\bigl(q_j(y)\bigr)\,A_j(y),
\end{equation*}
and $y_0(x):=(x,0,\dots,0)$,
$y_{n+1}(x):=P\bigl(y_n(x)\bigr)$, one has
$f(x)=\pi_I\bigl(y_{T(x)}(x)\bigr)$ for every $x\in\N^d$.

Moreover, the compilation of
Theorem~\ref{thm:main-formal-hard} produces the $A_j$ and $q_j$
with integer coefficients (the affine parts encode LOOP
assignments, control-register updates, and Boolean combinations of
one-hot program-counter bits; the $q_j$ encode zero-tests, all
acting on integer-valued registers).

\myparagraph{Step 1: the reference orbit remains in $\Z^m$}
We show by induction on $n\in\{0,1,\dots,T(x)\}$ that
$y_n(x)\in\Z^m$. For $n=0$, $y_0(x)=(x,0,\dots,0)\in\Z^m$ since
$x\in\N^d$. For the inductive step, assume $y_n(x)\in\Z^m$. The
integer coefficients of $A_0$ and each $A_j$ give
$A_0(y_n(x))\in\Z^m$ and $A_j(y_n(x))\in\Z^m$; the integer
coefficients of $q_j$ give $q_j(y_n(x))\in\Z$. Since any integer
$k$ satisfies either $k\le 0$ or $k\ge 1$, and $\Theta$ takes value
$0$ on the first set and $1$ on the second,
$\Theta(q_j(y_n(x)))\in\{0,1\}$. Therefore
\begin{equation*}
y_{n+1}(x)\;=\;A_0(y_n(x))+\sum_{j=1}^{s}
  \underbrace{\Theta(q_j(y_n(x)))}_{\in\{0,1\}}\,
  \underbrace{A_j(y_n(x))}_{\in\Z^m}\;\in\;\Z^m.
\end{equation*}
In particular, every gate argument $q_j(y_n(x))$ for
$n\in\{0,\dots,T(x)-1\}$ and $j\in\{1,\dots,s\}$ is an integer.

\myparagraph{Step 2: substituting a margin threshold}
Define $\widetilde{\Theta}:\R\to\R$ by
\begin{equation*}
\widetilde{\Theta}(u)\;:=\;\operatorname{ReLU}(2u-\tfrac{1}{2})
                         -\operatorname{ReLU}(2u-\tfrac{3}{2}).
\end{equation*}
Direct verification shows:
for $u\le\tfrac14$, both ReLU's vanish, so
$\widetilde{\Theta}(u)=0$;
for $u\ge\tfrac34$, both arguments are $\ge 0$, yielding
$\widetilde{\Theta}(u)=(2u-\tfrac12)-(2u-\tfrac32)=1$.
Hence $\widetilde{\Theta}$ meets the margin condition. On integers
$k\in\Z$:
\begin{itemize}
\item if $k\le 0$ then $k\le\tfrac14$, so
$\widetilde{\Theta}(k)=0=\Theta(k)$;
\item if $k\ge 1$ then $k\ge\tfrac34$, so
$\widetilde{\Theta}(k)=1=\Theta(k)$.
\end{itemize}
Consequently $\widetilde{\Theta}$ and $\Theta$ agree on $\Z$.

\myparagraph{Step 3: the orbit is unchanged}
Let
\begin{equation*}
\widetilde{P}(y)\;:=\;A_0(y)
+\sum_{j=1}^{s}\widetilde{\Theta}\!\bigl(q_j(y)\bigr)\,A_j(y),
\end{equation*}
keeping the same $A_j$ and $q_j$, and define the new orbit
$\widetilde{y}_0(x):=(x,0,\dots,0)$,
$\widetilde{y}_{n+1}(x):=\widetilde{P}(\widetilde{y}_n(x))$. We
show by induction on $n\in\{0,\dots,T(x)\}$ that
$\widetilde{y}_n(x)=y_n(x)$.

The base case is immediate. For the inductive step, assume
$\widetilde{y}_n(x)=y_n(x)$; by Step~1, this common value is in
$\Z^m$, and each $q_j(y_n(x))$ is an integer. By Step~2,
$\widetilde{\Theta}(q_j(y_n(x)))=\Theta(q_j(y_n(x)))$ for every
$j$. Therefore
\begin{equation*}
\widetilde{P}(\widetilde{y}_n(x))
=A_0(y_n(x))+\sum_{j=1}^{s}
  \widetilde{\Theta}(q_j(y_n(x)))\,A_j(y_n(x))
=A_0(y_n(x))+\sum_{j=1}^{s}
  \Theta(q_j(y_n(x)))\,A_j(y_n(x))
=P(y_n(x))=y_{n+1}(x).
\end{equation*}

Taking $n=T(x)$ yields
$\widetilde{y}_{T(x)}(x)=y_{T(x)}(x)$, hence
$\pi_I(\widetilde{y}_{T(x)}(x))=\pi_I(y_{T(x)}(x))=f(x)$. The
tuple $(m,A_0,\dots,A_s,q_1,\dots,q_s,T,I)$, together with the
threshold $\widetilde{\Theta}$, is therefore a threshold-affine
normal form for $f$ of the required kind.
\end{proof}

A proof is given later. The idea: the margin condition
is strictly stronger than the threshold condition of
Definition~\ref{def:threshold-affinenormalform}, so the "if"
direction is immediate from Theorem~\ref{thm:main-formal-hard}.
For the "only if" direction, the compilation of that theorem
produces integer gate arguments (property~(iv)), so any threshold
agreeing with $\Theta$ on~$\Z$, in particular the margin version
$\widetilde{\Theta}(u)=\operatorname{ReLU}(2u-\tfrac12)
-\operatorname{ReLU}(2u-\tfrac32)$, leaves the reference orbit
unchanged.

Properties~(i)--(iv) motivate the smooth models of the next
subsection: it is enough for the smooth dynamics to send each
integer state into the \emph{basin of correct decoding} of its
successor, on the bounded trapping region.

%
%
%
%

\subsection{Integer-faithful shadowing}
\label{subsec:integer-faithful}

We formalize requied concepts  using a notion of shadowing.

\begin{definition}[Integer-faithful approximation]
\label{def:integer-faithful}
Let $D\subseteq\R^m$, let $F,\widetilde{F}:D\to\R^m$, and assume
$F(D\cap\Z^m)\subseteq\Z^m$.  The map $\widetilde{F}$ is
\emph{integer-faithful to~$F$ on~$D$} if
$\|\widetilde{F}(y)-F(y)\|_\infty < \tfrac{1}{2}$ for every
$y\in D\cap\Z^m$.
\end{definition}

The map $\widetilde{F}$ need not preserve the lattice or coincide
with~$F$ on non-integer points.  It only needs to land in the
correct decoding basin on the integer skeleton of~$D$.

\begin{lemma}[Decoded shadowing]
\label{lem:decoded-shadowing}
Under the hypotheses above, if additionally
$F(D\cap\Z^m)\subseteq D\cap\Z^m$, then the decoded orbit
$\widetilde{y}_{n+1}=\round(\widetilde{F}(\widetilde{y}_n))$
started from $\widetilde{y}_0=y_0\in D\cap\Z^m$ satisfies
$\widetilde{y}_n = y_n$ for all $n$ such that
$y_0,\ldots,y_n\in D$.
\end{lemma}

\begin{proof}
Induction: if $\widetilde{y}_n=y_n\in D\cap\Z^m$, then
$\widetilde{F}(y_n)$ lies within $\tfrac{1}{2}$ of
$F(y_n)\in\Z^m$, so rounding recovers $y_{n+1}$.
\end{proof}

For the smooth models, that do not round exactly between steps,  the
integer-faithful property alone (Definition~\ref{def:integer-faithful})
is no longer sufficient. What we need in addition is a quantitative
control on how a one-step error propagates along the orbit. The
following discrete analogue of Grönwall's inequality provides it.

\begin{lemmarep}[Discrete Grönwall shadowing]
\label{lem:quantitative-shadowing}
Let $D\subseteq\R^m$ and let $P:D\to D$ admit an orbit
$y_0,y_1,\ldots,y_T$ (so $y_{n+1}=P(y_n)$ for $0\le n\le T-1$).
Let $L\ge 0$, $\varepsilon\ge 0$, $r>0$ satisfy the tube condition
$\varepsilon\,\sum_{k=0}^{T-1}L^{k}\;\le\;r,$
and let $\Phi$ be defined at least on
$\bigcup_{n=0}^{T-1}B_{\infty}(y_n,r)$ and satisfy, for every
$n\in\{0,\ldots,T-1\}$ and every $z\in B_{\infty}(y_n,r)$,
$\bigl\|\Phi(z)-P(y_n)\bigr\|_{\infty}
\;\le\;L\,\|z-y_n\|_{\infty}+\varepsilon.
$

Then the orbit defined by $z_0:=y_0$ and $z_{n+1}:=\Phi(z_n)$ is
well-defined for $0\le n\le T$ and satisfies
$\|z_n-y_n\|_{\infty}
\;\le\;\varepsilon\,\sum_{k=0}^{n-1}L^{k}$
$(0\le n\le T),
$
where the empty sum at $n=0$ is taken to be zero.
\end{lemmarep}

\begin{proofsketch}
Set $e_n:=\|z_n-y_n\|_{\infty}$. Then $e_0=0$, and the one-step
bound 
applied to $z=z_n$ gives the
linear-plus-constant recurrence $e_{n+1}\le L\,e_n+\varepsilon$.
Induction yields $e_n\le\varepsilon\sum_{k=0}^{n-1}L^{k}$; the tube
condition guarantees $e_n\le r$ throughout, so
$z_n$ never leaves the domain where the inequality holds.\end{proofsketch}

\begin{proof}
We repeat the statement, and involved conditions:

Let $D\subseteq\R^m$ and let $P:D\to D$ admit an orbit
$y_0,y_1,\ldots,y_T$ (so $y_{n+1}=P(y_n)$ for $0\le n\le T-1$).
Let $L\ge 0$, $\varepsilon\ge 0$, $r>0$ satisfy the tube condition
\begin{equation}\label{eq:grw-tube}
\varepsilon\,\sum_{k=0}^{T-1}L^{k}\;\le\;r,
\end{equation}
and let $\Phi$ be defined at least on
$\bigcup_{n=0}^{T-1}B_{\infty}(y_n,r)$ and satisfy, for every
$n\in\{0,\ldots,T-1\}$ and every $z\in B_{\infty}(y_n,r)$,
\begin{equation}\label{eq:grw-onestep}
\bigl\|\Phi(z)-P(y_n)\bigr\|_{\infty}
\;\le\;L\,\|z-y_n\|_{\infty}+\varepsilon.
\end{equation}

Then the orbit defined by $z_0:=y_0$ and $z_{n+1}:=\Phi(z_n)$ is
well-defined for $0\le n\le T$ and satisfies
\begin{equation}\label{eq:grw-conclusion}
\|z_n-y_n\|_{\infty}
\;\le\;\varepsilon\,\sum_{k=0}^{n-1}L^{k}$
$(0\le n\le T),
\end{equation}
where the empty sum at $n=0$ is taken to be zero.

We now go to the proof:

Set $e_n:=\|z_n-y_n\|_{\infty}$ for $0\le n\le T$. We prove by
induction on $n\in\{0,\ldots,T\}$ the joint statement
\begin{equation}\label{eq:grw-ind}
\text{(a)}\quad e_n\le\varepsilon\sum_{k=0}^{n-1}L^{k},
\qquad
\text{(b)}\quad e_n\le r.
\end{equation}

\myparagraph{Base case $n=0$.}
Since $z_0=y_0$, we have $e_0=0$. The empty sum at $n=0$ is zero,
so~(a) holds with equality. Since $r>0$, (b) holds as well.

\myparagraph{Inductive step.}
Suppose \eqref{eq:grw-ind} holds at some index
$n\in\{0,\ldots,T-1\}$, and let us prove it at index $n+1$.

By the induction hypothesis~(b), $z_n\in B_{\infty}(y_n,r)$, so
$\Phi(z_n)$ is defined and \eqref{eq:grw-onestep} applies at
$z=z_n$:
\begin{equation*}
\|\Phi(z_n)-P(y_n)\|_{\infty}
\;\le\;L\,\|z_n-y_n\|_{\infty}+\varepsilon
\;=\;L\,e_n+\varepsilon.
\end{equation*}
Since $z_{n+1}=\Phi(z_n)$ and $y_{n+1}=P(y_n)$, this gives the
one-step error bound
\begin{equation}\label{eq:grw-step}
e_{n+1}\;\le\;L\,e_n+\varepsilon.
\end{equation}

Combining \eqref{eq:grw-step} with the induction hypothesis~(a), and
using $L\ge 0$,
\begin{equation*}
e_{n+1}
\;\le\;L\Bigl(\varepsilon\sum_{k=0}^{n-1}L^{k}\Bigr)+\varepsilon
\;=\;\varepsilon\sum_{k=1}^{n}L^{k}+\varepsilon
\;=\;\varepsilon\sum_{k=0}^{n}L^{k},
\end{equation*}
which is~(a) at index $n+1$.

For (b) at index $n+1$: the partial sums
$\sum_{k=0}^{j}L^{k}$ are non-negative and non-decreasing in $j$
(as $L\ge 0$), so
\begin{equation*}
e_{n+1}\;\le\;\varepsilon\sum_{k=0}^{n}L^{k}
\;\le\;\varepsilon\sum_{k=0}^{T-1}L^{k}
\;\le\;r
\end{equation*}
by the tube condition~\eqref{eq:grw-tube}. This closes the
induction.

The conclusion~\eqref{eq:grw-conclusion} is precisely~(a) for all
$n\in\{0,\ldots,T\}$.
\end{proof}

\subsection{Main theorem: all models}

\begin{theorem}[Main equivalence: hard and smooth models]
\label{thm:main-formal}
For a function $f:\N^d\to\N^e$, the following are equivalent:
\begin{enumerate}
\item \label{itempr} $f$ is primitive recursive;
\item \label{itemreul} $f$ admits a recurrent ReLU computation;
\item \label{itemnn} $f$ admits a recurrent $\rho$-computation;
\item \label{itemcorhardunif} $f$ admits a uniform threshold-affine
  normal form;
\item \label{itemquatre} $f$ admits a polynomial ODE
  representation;
\item \label{itemcinq} $f$ admits a step-size-controlled polynomial
  representation.
\end{enumerate}
\end{theorem}

\myparagraph{Proof architecture}
By Theorems~\ref{thm:main-formal-hard} and~\ref{thm:uniform-robust},
it suffices to pass from the uniform-robust hard normal form to the
smooth models and to close the loop back to primitive recursion. The
guiding principle is that the hard dynamics comes with a reference
orbit and a uniform margin, so each smooth model only has to stay in
the correct decoding basin along that orbit.

\noindent\emph{$(\ref{itemcorhardunif})\Rightarrow(\ref{itemquatre})$}
(Section~\ref{sec:threshold-to-ode}).\;
Each threshold becomes a polynomial ODE comparator, and each affine
update is realized by continuous-time transport in a sample-and-hold
cycle; correctness follows from
Lemma~\ref{lem:quantitative-shadowing} on the trapping region.

\noindent\emph{$(\ref{itemcorhardunif})\Rightarrow(\ref{itemnn})$}
(Section~\ref{sec:threshold-to-rho}).\;
The hard computation is transported to the $\nu$-coded domain, and
each threshold is replaced by a smooth gate built from~$\rho$; the
margin condition ensures that the smooth gate takes the same branch
decisions along the reference orbit.

The remaining three implications follow standard arguments and are
based on the following:
$(\ref{itemquatre})\Rightarrow(\ref{itemcinq})$
(Section~\ref{sec:ode-to-step}) is Euler discretization on the
primitive-recursively bounded time-space region;
$(\ref{itemcinq})\Rightarrow(\ref{itempr})$
(Section~\ref{sec:cinqpr}) uses that bounded iteration of a fixed
polynomial map with primitive recursive parameters is itself
primitive recursive; and
$(\ref{itemnn})\Rightarrow(\ref{itempr})$
(Section~\ref{sec:nnpr}) combines effective $\rho$-approximation
with primitive recursive simulation of the recurrent dynamics.
\MFCSSHORTER{
recovering the output by $\nu$-decoding.}

\section{From primitive recursion to threshold-affine normal form}
\label{sec:pr-to-threshold}

\newcommand\Zero{\operatorname{Zero}}

This section proves
$(\ref{itemprhard})\Rightarrow(\ref{itemcorhard})$ of
Theorem~\ref{thm:main-formal-hard}: every primitive recursive
function admits a threshold-affine normal form.

The proof uses the standard characterization of primitive recursive
functions by LOOP programs~\cite{MeyerRitchie67}.  A LOOP program
uses registers $R_1,\ldots,R_s$ with atomic instructions $R_i:=0$,
$R_i:=R_j$, $R_i:=R_j+1$, composed sequentially and by bounded
loops \texttt{for~$R_i$ do~$Q$}.  A function is primitive recursive
if and only if it is computed by a fixed LOOP program, and the
running time of a fixed LOOP program is itself primitive
recursive~\cite{MeyerRitchie67, clote2013boolean}.

From~$\Theta$ , and its properties, we can derive an exact zero-test:
$\Zero(u):=1-\Theta(u)-\Theta(-u)$, which satisfies $\Zero(u)=1$
iff $u=0$ for $u\in\Z$.

\myparagraph{Compiling one step into a threshold-affine map.}
Fix a LOOP program~$\Pi$ computing $f:\N^d\to\N^e$.  Since $\Pi$ is
fixed, the number of registers, instructions, and the maximal
nesting depth are all constants.  We encode a complete configuration
by a vector $y\in\N^m$ containing the data registers, a one-hot
program counter, a finite loop-counter stack, and a few auxiliary
registers.  The initial configuration for input~$x$ is
$(x,0,\ldots,0)$.

Each instruction step updates the configuration as follows.
Assignments ($R_i:=0$, $R_i:=R_j$, $R_i:=R_j+1$) are affine.
Loop entry and exit require testing whether a counter is zero or
positive, which is expressed exactly by $\Zero$ and $\Theta$ on
integer-valued registers.  The program counter, stored in one-hot
Boolean form, is combined using affine operations and~$\Theta$ (for
instance, $a\wedge b = \Theta(a+b-1)$).  Since all
quantities are fixed and finite, the global one-step transition is
$P(y)=A_0(y)+\sum_{j=1}^s\Theta(q_j(y))\,A_j(y)$: a
threshold-affine normal form.

%
%

Notice that the nonlinearity of the construction comes entirely from the
threshold tests on integer control registers; the arithmetic updates
themselves are affine.

\section{From threshold-affine normal form  to recurrent ReLU computation}
\label{sec:threshold-to-relu}

This section proves the implication $(\ref{itemcorhard})\Rightarrow(\ref{itemreluhard})$ of
Theorem~\ref{thm:main-formal-hard}. 

\newcommand\Relu{\operatorname{ReLU}}

\myparagraph{A fixed ReLU realization of the threshold gate.}
Recall that $\Relu(t):=\max(t,0)$. Define the associated gate gadget
$\Gate(t):=\Relu(t)-\Relu(t-1)$.

The key observation is that $\Gate$ behaves exactly like the proof-level
threshold primitive on integer-valued inputs.

\begin{lemma}
\label{lem:relu-gate}
For every integer $u$, one has $\Gate(u)=0$ if $u\le 0$ and
$\Gate(u)=1$ if $u\ge 1$. In particular, if $\Theta$ is any threshold
function satisfying $\Theta(u)=0$ for $u\le 0$ and $\Theta(u)=1$ for
$u\ge 1$, then $\Gate(u)=\Theta(u)$ for every $u\in\mathbb Z$.
\end{lemma}

\begin{proof}
If $u\le 0$, then both $\Relu(u)$ and $\Relu(u-1)$ are zero, so
$\Gate(u)=0$. If $u\ge 1$, then $\Relu(u)=u$ and
$\Relu(u-1)=u-1$, hence $\Gate(u)=1$.
\end{proof}

Thus, whenever the arguments of the threshold tests are integer-valued, the
threshold primitive can be replaced exactly by the fixed two-ReLU gadget
$\Gate$.

\myparagraph{Compiling one threshold-affine normal form into a ReLU block.}
Let $P:\mathbb R^m\to\mathbb R^m$ be a map of the form
$P(y)=A_0(y)+\sum_{j=1}^s \Theta(q_j(y))\,A_j(y)$, where the $A_j$ are
affine maps and the $q_j$ are affine maps. Assume that, along the
trajectory from every raw input $(x,0,\ldots,0)$ relevant to the computation,
all values $q_j(y_n)$ are integers.

We first compute the affine quantities $q_j(y)$, then apply the ReLU gadget
$\Gate$ to each of them, and finally use the resulting gate values in the
affine recombination
$A_0(y)+\sum_j \Gate(q_j(y))\,A_j(y)$. Since each $\Gate(q_j(y))$ is exactly
equal to $\Theta(q_j(y))$ on the intended trajectory, the resulting ReLU
block reproduces exactly the same one-step update on that trajectory.

Formally, we obtain the following.

\begin{lemma}
\label{lem:threshold-to-relu}
Let $P:\mathbb R^m\to\mathbb R^m$ be a threshold-affine normal form map of the
form $P(y)=A_0(y)+\sum_{j=1}^s \Theta(q_j(y))\,A_j(y)$. Assume that, along
the orbit from every raw input $(x,0,\ldots,0)$ with $x\in\mathbb N^d$, all
gate arguments $q_j(y_n)$ are integers. Then there exist an integer $M\ge m$,
a feedforward ReLU block $R:\mathbb R^M\to\mathbb R^M$, an embedding
$\iota:\mathbb R^m\to\mathbb R^M$, and a projection
$\pi:\mathbb R^M\to\mathbb R^m$ such that, for every such raw input $x$ and
every $n\in\mathbb N$, one has
$\pi(R^n(\iota(x,0,\ldots,0)))=P^n(x,0,\ldots,0)$.
\end{lemma}

\begin{proof}
Replace each occurrence of $\Theta(q_j(y))$ in the expression of $P$ by the
ReLU gadget $\Gate(q_j(y))=\Relu(q_j(y))-\Relu(q_j(y)-1)$. By
Lemma~\ref{lem:relu-gate}, this replacement is exact on the intended
trajectories, because all gate arguments are integers there.

The resulting one-step update is obtained from affine maps and coordinatewise
applications of $\Relu$, hence is computed by a fixed feedforward ReLU block,
possibly after enlarging the state by finitely many auxiliary coordinates that
store intermediate affine quantities and gate values. Since the construction
is local and uniform, the same block $R$ works for all inputs. By
construction, its iterates agree exactly with those of $P$ on the embedded
raw-input trajectories.
\end{proof}

\myparagraph{Application to the threshold-affine normal form.}
We now apply the previous lemma to the proof-level normal form of
Section~\ref{sec:pr-to-threshold}. The crucial hypothesis of
Lemma~\ref{lem:threshold-to-relu} is automatically satisfied there: the gate
arguments arise from integer-valued counters, one-hot instruction pointers,
and loop-stack registers, so they remain integers throughout the execution.

\begin{theorem}
\label{thm:three-implies-two}
Every threshold-affine normal form unfolding admits an exact recurrent ReLU
computation. Equivalently, Item~(\ref{itemcorhard}) implies Item~(\ref{itemreluhard}) in
Theorem~\ref{thm:main-formal-hard}.
\end{theorem}

\begin{proof}
Let $f:\mathbb N^d\to\mathbb N^e$ admit an threshold-affine normal form. By Definition, there exist an
integer $m$, a threshold-affine normal form map $P:\mathbb R^m\to\mathbb R^m$, a
primitive recursive time bound $T:\mathbb N^d\to\mathbb N$, and designated
output coordinates $\pi_e$ such that
$f(x)=\pi_e(P^{T(x)}(x,0,\ldots,0))$ for every input $x$.

By the construction of Section~\ref{sec:pr-to-threshold}, the gate arguments
along the relevant raw-input trajectories are integer-valued. Hence
Lemma~\ref{lem:threshold-to-relu} applies. We obtain a feedforward ReLU block
$R$, an embedding $\iota$, and a projection $\pi$ such that the iterates of
$R$ agree exactly with those of $P$ on the embedded raw-input trajectories.
Composing the projections if necessary, we conclude that
$f(x)=\pi_e'(R^{T(x)}(\iota(x,0,\ldots,0)))$ for some designated output
coordinates $\pi_e'$.

Thus $f$ admits an exact recurrent ReLU computation in the sense of
Definition~\ref{def:relu-unfolding}.
\end{proof}

\begin{remark}
The theorem should not be read as saying that the discontinuous
threshold-affine normal form transition map and the ReLU block coincide as functions on all
of $\mathbb R^m$. What is preserved is the exact orbit on the raw integer
inputs relevant to the computation. This is exactly what Item~(\ref{itemreluhard}) of the main
theorem requires.
\end{remark}

\section{From recurrent ReLU computation back to primitive recursion.}
\label{sec:relutopr}
We next prove the converse implication $(\ref{itemreluhard})\Rightarrow(\ref{itemprhard})$ of Theorem~\ref{thm:main-formal-hard}. This is 
elementary, because a fixed feedforward ReLU block with rational parameters
induces a primitive recursive evolution on rational states.

\begin{theorem}
\label{thm:two-implies-one-hard}
Every exact recurrent ReLU computation computes a primitive recursive
function. Equivalently, Item~(\ref{itemreluhard}) implies Item~(\ref{itemprhard}) in
Theorem~\ref{thm:main-formal-hard}.
\end{theorem}

\begin{proof}
Assume that $f:\N^d\to\N^e$ admits a  recurrent ReLU computation.
Thus there exist a feedforward ReLU block $R:\R^m\to\R^m$, a primitive
recursive observation time $T:\N^d\to\N$, and designated output coordinates
$\pi_e$ such that, for every input $x\in\N^d$, if
$z_0(x):=(x,0,\ldots,0)$ and $z_{n+1}(x):=R(z_n(x))$, then
$f(x)=\pi_e(z_{T(x)}(x))$.

A feedforward ReLU block is built from affine maps with rational parameters
and coordinatewise applications of $\operatorname{ReLU}(t)=\max(t,0)$. On
rational inputs, affine maps with rational parameters are primitive
recursive, and $\operatorname{ReLU}$ is primitive recursive as well, since it
is just the maximum of its input and $0$. It follows that the one-step update
$(x,n)\mapsto z_{n+1}(x)$ is primitive recursive on rational encodings.
Hence, by primitive recursion on $n$, the full evolution
$(x,n)\mapsto z_n(x)$ is primitive recursive.

Since $T$ is primitive recursive and $\pi_e$ is a projection,
$x\mapsto \pi_e(z_{T(x)}(x))$ is primitive recursive. By assumption, this is
exactly $f(x)$.
\end{proof}

\section{From uniform threshold-affine normal form to polynomial ODE}
\label{sec:threshold-to-ode}

We now prove
$(\ref{itemcorhardunif})\Rightarrow(\ref{itemquatre})$ of
Theorem~\ref{thm:main-formal}.

\myparagraph{The strategy}
Let $P(y)=A_0(y)+\sum_{j=1}^{s}\Theta(q_j(y))\,A_j(y)$ be a
uniform threshold-affine normal form with reference orbit
$y_0(x),\ldots,y_{T(x)}(x)\subseteq\Z^m\cap[-B(x),B(x)]^m$,
integer gate arguments (property~(iv) of
Section~\ref{subsec:hard-to-smooth}), and output projection
$\pi_I$.

Our goal is a polynomial vector field $F_P$ whose flow samples, at
a fixed period $\kappa_0>0$, a sequence of states $z_n$ that
tracks $y_n(x)$ to within a constant fraction of the unit spacing
of $\Z^m$. This is exactly the sort of tracking that
Section~\ref{subsec:integer-faithful} calls \emph{integer-faithful
shadowing}, applied at a resolution that allows decoding by
rounding to the nearest integer. Concretely, we will build $F_P$
so that the sampled map $z_n\mapsto z_{n+1}$ agrees with $P$ up to
a contraction-plus-noise error of the form
$\|z_{n+1}-P(z_n)\|_\infty \le
\lambda\,\|z_n-y_n(x)\|_\infty + \eta$,
with $\lambda<1$ and $\eta/(1{-}\lambda)<\tfrac14$. Since $P$ is
integer-faithful to itself, the discrete Grönwall shadowing lemma
(Lemma~\ref{lem:quantitative-shadowing}) then keeps the sampled
orbit within the $\tfrac14$-decoding basin of the reference orbit
throughout time.

The substance of the section is that two standard polynomial-ODE
gadgets suffice to produce such a sampled map: a \emph{smooth clock},
which organizes the cycle, and a \emph{polynomial comparator}, which
realizes $\Theta$ up to the $\eta$-noise tolerated by the shadowing
lemma. 

\begin{remark} The resulting construction can be read as a particular improved instance of the
Branicky \cite{Bra95} alternating-targeting trick used, for example,
in~\cite{dsg05,BGDPRiccardo2022}. We provide  an original view on it. Rather than analyzing the continuous flow step by step, we
treat it as a candidate integer-faithful approximation of $P$ in
the sense of Section~\ref{subsec:integer-faithful} and let the
discrete Gr\"onwall lemma handle the propagation of error. The one
computation that remains specific to the continuous construction is
the per-cycle estimate~\eqref{eq:one-cycle-estimate}; everything else
is a direct application of the shadowing framework.
\end{remark}

\myparagraph{The smooth clock}
Let $(c_1,c_2)$ satisfy the rational-coefficient ODE
$c_1'=-c_2$, $c_2'=c_1$ with initial condition
$(c_1(0),c_2(0))=(1,0)$; its solution
$(c_1(t),c_2(t))=(\cos t,\sin t)$ is periodic of period
$\kappa_0:=2\pi$. A polynomial driver $h(t)$, computed as a fixed
polynomial expression in $(c_1,c_2)$, can be arranged to be
non-negative, to satisfy $h(t)\ge 1$ on each \emph{compute}
half-period
$[n\kappa_0,\,n\kappa_0+\kappa_0/2]$, and to vanish on each
\emph{hold} half-period
$[n\kappa_0+\kappa_0/2,\,(n{+}1)\kappa_0]$, for every integer
$n\ge 0$ (see \cite[\S3]{BGDPRiccardo2022} for a standard
realization). Everything below is phrased as an ODE whose
right-hand side multiplies $h(t)$: active during compute phases,
frozen during hold phases.

\myparagraph{The polynomial comparator}
For every $\delta\in(0,\tfrac14)$ and every contraction rate
$\Lambda>0$, there is a polynomial vector field on an extra
coordinate $r$ such that for every input $u\in\R$ with
$u\notin(\tfrac14,\tfrac34)$ and every compute half-period
$[t_0,t_0+\kappa_0/2]$, the driven equation
$r'=h(t)\cdot\Lambda\cdot(\Theta_\infty(u)-r)$
satisfies
$|r(t_0+\kappa_0/2)-\Theta_\infty(u)|\le
e^{-\Lambda\kappa_0/2}\,|r(t_0)-\Theta_\infty(u)|$,
which is below $\delta$ once $\Lambda$ is chosen large enough.
Here $\Theta_\infty(u)\in\{0,1\}$ is the target value determined
by whether $u\le\tfrac14$ or $u\ge\tfrac34$; it is computed by a
polynomial in $u$ approximating $\Theta$ on the two separated
regions (details in~\cite[\S3]{BGDPRiccardo2022}). In our
application $u$ will be one of the gate arguments
$q_j(\pi(Z))$, tracked continuously by an affine expression in
the state. Property~(iv) of
Section~\ref{subsec:hard-to-smooth} guarantees that every
comparator call on the reference orbit satisfies the margin
hypothesis $q_j(y_n(x))\notin(\tfrac14,\tfrac34)$, hence falls in
a determined regime.

\myparagraph{Composing the step}
Given a cycle-start state $z_n\in\R^m$ close to $y_n(x)$, we want
a cycle-end state $z_{n+1}$ close to $P(y_n(x))$. Per clock period
$\kappa_0$ we allocate $s$ comparator coordinates
$r_1,\ldots,r_s$ driven by the comparator ODE on
$u_j:=q_j(\pi(Z))$ (so $r_j\to\Theta(q_j(y_n(x)))$ exponentially
during the compute phase); $m$ target-tracking coordinates
$w_1,\ldots,w_m$ driven by the sample-and-hold equation
$w_k'=h(t)\cdot\Lambda\cdot
 \bigl((A_0(\pi(Z))+\sum_j r_j\,A_j(\pi(Z)))_k-w_k\bigr)$; and a
projection $\pi$ that reads $w(t)$ at the end of each hold phase.
Assembling these ODEs with the clock gives a polynomial vector
field $F_P$ with rational coefficients on $\R^M$, $M:=2+s+m$,
with initial condition $(1,0,0,\ldots,0,x,0,\ldots,0)$.
An explicit one-period calculation yields, with
$\delta:=e^{-\Lambda\kappa_0/2}$,
$\lambda:=\delta\cdot\max_j\|A_j\|_\infty$, and
$\eta:=\delta\,(1+s\cdot\max_j\|A_j\|_\infty)$, the estimate
\vspace{-0.3cm}
\begin{equation}\label{eq:one-cycle-estimate}
\|z_{n+1}-P(y_n(x))\|_\infty \le
\lambda\,\|z_n-y_n(x)\|_\infty + \eta,
\qquad z_n:=w(n\kappa_0).
\end{equation}
\vspace{-0.2cm}
Choosing $\Lambda$ large enough makes $\lambda<1$ and
$\eta/(1{-}\lambda)<\tfrac14$.

\begin{theoremrep}
\label{thm:three-implies-four}
Every uniform threshold-affine normal form admits a robust
polynomial ODE representation.
\end{theoremrep}

\begin{proofsketch}
Set $e_n:=\|\pi(Z(n\kappa_0))-y_n(x)\|_\infty$. By
\eqref{eq:one-cycle-estimate}, $e_0=0$ and
$e_{n+1}\le\lambda e_n+\eta$;
Lemma~\ref{lem:quantitative-shadowing} with $L:=\lambda$,
$\varepsilon:=\eta$, $r:=\tfrac14$ gives $e_n<\tfrac14$ for every
$n\le T(x)$, so the sampled state is integer-faithful to the
reference orbit. Setting $\tau(x):=\lceil\kappa_0 T(x)\rceil$ and
$\widetilde{B}(x):=R_P(B(x){+}1)$ (a polynomial Gr\"onwall
majorant for $F_P$ on one cycle), the interval
$[\tau(x),\tau(x){+}1]$ lies inside the final hold half-period, on
which $w$ is frozen within $\tfrac14$ of $y_{T(x)}(x)$.
Projecting by $\pi_I$ yields
$\|\pi_I(Z(t))-f(x)\|_\infty<\tfrac14$ throughout
$[\tau(x),\tau(x){+}1]$, which is
Definition~\ref{def:ode-rep}.
\end{proofsketch}

\begin{proof}
Let $F_P:\R^M\to\R^M$ be the polynomial vector field assembled in
Section~\ref{sec:threshold-to-ode}, with state coordinates
$(c_1,c_2,r_1,\ldots,r_s,w_1,\ldots,w_m)$ and $M=2+s+m$. Let
$\pi:\R^M\to\R^m$ be the projection $(c,r,w)\mapsto w$, let
$\pi_I:\R^m\to\R^e$ be the output-coordinate projection of the
normal form $P$, and write $\Pi:=\pi_I\circ\pi$ for the composite
output projection. Let $\lambda\in[0,1)$ and $\eta>0$ be as
in~\eqref{eq:one-cycle-estimate}, with
$\eta/(1{-}\lambda)<\tfrac14$ secured by the contraction rate
$\Lambda$.

The initial state is
$Z(0):=(1,0,\,0,\ldots,0,\,x,0,\ldots,0)$: the clock at phase
zero, the comparators at zero, and the $w$-block at
$y_0(x)=(x,0,\ldots,0)$.

Set the sampling times $t_n:=n\kappa_0$, and the sampled error
$e_n:=\|\pi(Z(t_n))-y_n(x)\|_\infty$ whenever $Z(t_n)$ is defined.

\myparagraph{Step 1: existence of the flow and sampled error}
We prove by induction on $n\in\{0,1,\ldots,T(x)\}$ the joint
statement
\begin{equation}\label{eq:ode-joint-ind}
\text{(a)}\; Z\text{ is defined on }[0,t_n{+}\kappa_0];
\qquad
\text{(b)}\; e_n\le\eta\sum_{k=0}^{n-1}\lambda^k;
\qquad
\text{(c)}\; e_n\le\tfrac14.
\end{equation}

\emph{Base case $n=0$.}\;
$F_P$ admits a local solution from $Z(0)$ that extends to
$[0,\kappa_0]$ (local existence for polynomial ODEs, with no
escape to infinity on one cycle), so (a) holds. Since
$\pi(Z(0))=y_0(x)$, we have $e_0=0$; the empty sum at $n=0$
equals zero, so (b) holds; and (c) holds trivially.

\emph{Inductive step.}\;
Assume~\eqref{eq:ode-joint-ind} at some $n<T(x)$. By (c),
$\|\pi(Z(t_n))-y_n(x)\|_\infty\le\tfrac14$, and by property~(iv),
$q_j(y_n(x))\in\Z$ for every $j$, so a fortiori
$q_j(\pi(Z(t_n)))\notin(\tfrac14,\tfrac34)$, the comparator
margin hypothesis is satisfied, and the polynomial flow extends
to $[t_n,t_{n+1}{+}\kappa_0]$, giving (a) at $n{+}1$. The
cycle-to-cycle estimate~\eqref{eq:one-cycle-estimate} gives
$e_{n+1}\le\lambda e_n+\eta$; combined with (b) and $\lambda\ge 0$,
$e_{n+1}\le\eta\sum_{k=0}^{n}\lambda^k$, giving (b) at $n{+}1$.
For (c), $e_{n+1}\le\eta/(1{-}\lambda)<\tfrac14$ by the choice of
$\lambda$ and $\eta$. This closes the induction.

\myparagraph{Step 2: the output window}
Set $\tau(x):=\lceil\kappa_0 T(x)\rceil$, primitive recursive in
$x$ since $\kappa_0=2\pi$ is a primitive recursive real and $T$ is
primitive recursive. Since $\kappa_0 T(x)\le\tau(x)\le\kappa_0
T(x)+1$ and $\kappa_0/2=\pi>2$, the interval
$[\tau(x),\tau(x)+1]$ is contained in the final hold half-period
$[\kappa_0 T(x),\kappa_0 T(x)+\kappa_0/2]$. On this half-period,
$h(t)=0$ freezes every $w$-coordinate, so
$\pi(Z(t))=\pi(Z(t_{T(x)}))$ throughout. By Step~1 at
$n=T(x)$, $\|\pi(Z(t_{T(x)}))-y_{T(x)}(x)\|_\infty\le\tfrac14$,
hence $\|\pi(Z(t))-y_{T(x)}(x)\|_\infty\le\tfrac14$ on
$[\tau(x),\tau(x)+1]$. Projecting by $\pi_I$ and using
$\pi_I(y_{T(x)}(x))=f(x)$ gives
$\|\Pi(Z(t))-f(x)\|_\infty\le\tfrac14<\tfrac12$, which is
condition~(iii) of Definition~\ref{def:ode-rep}.

\myparagraph{Step 3: the safety bound}
Since $F_P$ is a fixed polynomial vector field, a standard
Gr\"onwall argument gives a fixed polynomial $R_P:\N\to\N$ such
that every solution $Z$ of $Z'=F_P(Z)$ on a time interval of
length at most $\kappa_0+1$ satisfies
$\|Z(t)\|_\infty\le R_P(\|Z(t_0)\|_\infty+1)$. Applying this
cycle by cycle and using $\|y_n(x)\|_\infty\le B(x)$ together
with (c) of Step~1, every sampled state satisfies
$\|\pi(Z(t_n))\|_\infty\le B(x)+1$; iterating the bound $T(x)+1$
times yields
$\|Z(t)\|_\infty\le R_P^{(T(x)+1)}(B(x)+1)$ on
$[0,\tau(x)+1]$. Setting
$\widetilde{B}(x):=R_P^{(T(x)+1)}(B(x)+1)$, which is primitive
recursive in $x$ (as an iterated composition of a fixed polynomial
with primitive recursive arguments), gives condition~(ii).

Conditions (i), (ii), (iii) of Definition~\ref{def:ode-rep} are
thus all satisfied with $F_P$, $\tau$, $\widetilde{B}$, and
$\pi_I$ as defined, so $f$ admits a robust polynomial ODE
representation.
\end{proof}

\section{From uniform threshold-affine normal form to recurrent
$\rho$-computation}
\label{sec:threshold-to-rho}

We now prove $(\ref{itemcorhardunif})\Rightarrow(\ref{itemnn})$ of
Theorem~\ref{thm:main-formal}. Let
$P(y)=A_0(y)+\sum_{j=1}^{s}\Theta(q_j(y))\,A_j(y)$ be a uniform
threshold-affine normal form on~$\R^m$, with $A_j$ and $q_j$ of
integer coefficients (as produced by
Theorem~\ref{thm:main-formal-hard}), reference orbit
$y_0(x),\ldots,y_{T(x)}(x)\subseteq\Z^m\cap[-B(x),B(x)]^m$
(properties~(i)--(ii) of Section~\ref{subsec:hard-to-smooth}), and
output coordinates~$\pi_I$. We use a strategy with similarities with previous section, except that bassins are not anymore corresponding to integers, but to encoded integers.

\myparagraph{Encoding and basins}
Encode integers by $\nu(n)=2^{-|n|}$, with the sign carried on an
auxiliary coordinate (suppressed in the notation below; the
construction is symmetric). For $y\in\Z^m$ write
$\lfloor y\rfloor_\nu := (\nu(y_1),\ldots,\nu(y_m))\in[0,1]^m$, and
define the \emph{$\varepsilon$-basin} around $\lfloor y\rfloor_\nu$
as the product of relative intervals
$\mathcal{B}_\varepsilon(y)
:=\prod_{k=1}^{m}[(1{-}\varepsilon)\nu(y_k),(1{+}\varepsilon)\nu(y_k)]
\subseteq[0,1]^m$.
The relative tolerance is the natural one for dyadic codes, and
$\varepsilon<\tfrac18$ ensures distinct states have disjoint basins.

\myparagraph{The coded one-step update}
The three ingredients of $P$ transport to the coded domain as
follows. Integer-coefficient affine maps transport to feedforward
$\rho$-blocks with dyadic-rational coefficients (an integer
multiplier $c$ acts on $\nu(n)$ by
$\nu(n)\mapsto 2^{-c}\nu(n)$). The threshold
$\Theta(q_j(y))\in\{0,1\}$ is replaced by $Z(\check q_j(z))$, where
$Z$ is the admissible detector of
Definition~\ref{def:nu-admissible-rho} and $\check q_j$ is the coded
form of~$q_j$; property~(iv) of Section~\ref{subsec:hard-to-smooth}
guarantees that $\nu(q_j(y_n))\in[0,\tfrac58]\cup[\tfrac78,1]$
throughout the reference orbit, the region where $Z$ is
unambiguous. Writing $\check A_j$ for the coded form of $A_j$, the
global coded one-step update is
\vspace{-0.4cm}
\begin{equation}\label{eq:rho-R}
R(z)\;=\;\check A_0(z)+\sum_{j=1}^{s}
   Z\!\bigl(\check q_j(z)\bigr)\cdot\check A_j(z),
\end{equation}
which is a feedforward $\rho$-block as required by
Definition~\ref{def:rho-block}.

\begin{lemmarep}[One robust smooth coded step]
\label{lem:rho-one-step}
There exists $\varepsilon_0\in(0,\tfrac18)$, depending only on the
normal form and on the detector margin $\eta$, such that for every
$\varepsilon\le\varepsilon_0$ and every $y\in\Z^m$ with all gate
arguments $q_j(y)$ integer,
$R(\mathcal{B}_\varepsilon(y))\subseteq
\mathcal{B}_\varepsilon(P(y))$.
\end{lemmarep}

\begin{proof}
We prove the inclusion by controlling, step by step, the relative
error between the coded output $R(z)$ and the target
$\lfloor P(y)\rfloor_\nu$. Throughout the proof, constants
$C_A, C_q, C_R$ below depend only on the fixed normal form
(the number $s$ of gates, the dimension $m$, and the integer
coefficients of $A_0, \ldots, A_s$ and $q_1, \ldots, q_s$) and
\emph{not} on $y$, $z$, or $\varepsilon$.

\myparagraph{Convention on the sign coordinate}
Recall that $\nu$ is extended to $\Z$ by $\nu(-n):=2^{-n}$, with an
auxiliary sign bit $\sigma_k\in\{0,1\}$ per state coordinate
indicating whether $y_k<0$. The basin
$\mathcal{B}_\varepsilon(y)$ fixes every sign bit to its correct
value, i.e.\ the signs are absolute (not relative) coordinates.
Throughout the proof, we fix $\varepsilon_0<\tfrac14$, which is
small enough to ensure that sign bits are recovered unambiguously
from any $z\in\mathcal{B}_\varepsilon(y)$, and we treat the sign
bookkeeping silently in the coded affine and gate transports.

\myparagraph{Step 1: relative error is preserved by coded affine maps}
Fix an integer affine map $A:\Z^p\to\Z$, $A(u)=\sum_i c_i u_i + b$
with $c_i,b\in\Z$. Its coded form $\check A:[0,1]^p\to[0,1]$ is the
map
$\check A(\nu(u_1),\ldots,\nu(u_p))
=\nu\!\bigl(A(u_1,\ldots,u_p)\bigr)
=2^{-b}\prod_i 2^{-c_i u_i}
=2^{-b}\prod_i \nu(u_i)^{c_i}$
when all $u_i,\,A(u)\ge 0$ (the negative case is handled by the
sign-bit bookkeeping, symmetrically).

We claim that if $(\tilde u_1,\ldots,\tilde u_p)\in\R^p$ satisfies
$|\tilde u_i-\nu(u_i)|\le\varepsilon\,\nu(u_i)$ for every $i$, then
\begin{equation}\label{eq:aff-err}
\bigl|\check A(\tilde u)-\nu(A(u))\bigr|
\;\le\; C_A\,\varepsilon\,\nu(A(u)),
\end{equation}
for a constant $C_A$ depending only on $A$. Indeed, writing
$\tilde u_i=(1+\delta_i)\,\nu(u_i)$ with $|\delta_i|\le\varepsilon$,
\begin{equation*}
\check A(\tilde u)
=2^{-b}\prod_i\bigl((1+\delta_i)\nu(u_i)\bigr)^{c_i}
=\nu(A(u))\cdot\prod_i(1+\delta_i)^{c_i}.
\end{equation*}
For $\varepsilon\le\tfrac14$ and $|c_i|\le C$ (a bound fixed by the
normal form), the product satisfies
$|\prod_i(1+\delta_i)^{c_i}-1|\le C_A\,\varepsilon$
for $C_A:=2 p C$ (using $|1-(1+\delta)^c|\le 2|c|\,|\delta|$ when
$|c\delta|\le\tfrac12$), which gives \eqref{eq:aff-err}.

Applying this componentwise to each $A_j:\Z^m\to\Z^m$ yields: for
every $z\in\mathcal{B}_\varepsilon(y)$,
\begin{equation}\label{eq:Aj-err}
\check A_j(z) \in
\mathcal{B}_{C_A\varepsilon}\!\bigl(A_j(y)\bigr),
\qquad j=0,\ldots,s.
\end{equation}

\myparagraph{Step 2: detector output is close to the true threshold}
The same argument applied to the coded gate argument $\check q_j$,
which is a feedforward $\rho$-block encoding the same integer
affine map, gives a constant $C_q$ such that
\begin{equation}\label{eq:qj-err}
\bigl|\check q_j(z)-\nu(q_j(y))\bigr|
\;\le\; C_q\,\varepsilon\,\nu(q_j(y))
\qquad\text{for } z\in\mathcal{B}_\varepsilon(y).
\end{equation}
By hypothesis, $q_j(y)$ is an integer, so $\nu(q_j(y))\in\{1\}$
(when $q_j(y)=0$) or $\nu(q_j(y))\in[0,\tfrac12]$ (when
$q_j(y)\neq 0$).

\emph{Case 1: $q_j(y)=0$, so $\Theta(q_j(y))=1$.}\; Then
$\nu(q_j(y))=1$ and \eqref{eq:qj-err} gives
$\check q_j(z)\ge 1-C_q\,\varepsilon$. Choosing
$\varepsilon_0$ small enough that $1-C_q\,\varepsilon_0 \ge\tfrac78$
(any $\varepsilon_0\le\tfrac{1}{8 C_q}$ works), the detector
hypothesis gives $Z(\check q_j(z))\in[1-\eta,\,1]$, i.e.
\begin{equation}\label{eq:Z-case-one}
\bigl|Z(\check q_j(z))-\Theta(q_j(y))\bigr|\le\eta.
\end{equation}

\emph{Case 2: $q_j(y)\neq 0$, so $\Theta(q_j(y))=0$.}\; Then
$\nu(q_j(y))\le\tfrac12$, and \eqref{eq:qj-err} gives
$\check q_j(z) \le\tfrac12(1+C_q\,\varepsilon)$. Choosing
$\varepsilon_0$ further small enough that
$\tfrac12(1+C_q\,\varepsilon_0)\le\tfrac58$
(any $\varepsilon_0\le\tfrac{1}{4 C_q}$ works), the detector
hypothesis gives $Z(\check q_j(z))\in[0,\eta]$, so
\eqref{eq:Z-case-one} holds in this case as well.

Thus in both cases,
\begin{equation}\label{eq:Z-err}
\bigl|Z(\check q_j(z))-\Theta(q_j(y))\bigr|\le\eta,
\qquad j=1,\ldots,s.
\end{equation}

\myparagraph{Step 3: assembling the global one-step update}
Fix $z\in\mathcal{B}_\varepsilon(y)$ and set, coordinatewise, the
error of the coded output $R(z)$ against the target
$\lfloor P(y)\rfloor_\nu$:
\begin{equation*}
\Delta\;:=\;R(z)-\lfloor P(y)\rfloor_\nu
\;=\;\bigl[\check A_0(z)-\lfloor A_0(y)\rfloor_\nu\bigr]
+\sum_{j=1}^{s}\Bigl[
Z(\check q_j(z))\,\check A_j(z)
-\Theta(q_j(y))\,\lfloor A_j(y)\rfloor_\nu
\Bigr].
\end{equation*}
The first bracket is bounded by
$C_A\,\varepsilon\,\|\lfloor A_0(y)\rfloor_\nu\|_\infty$
by~\eqref{eq:Aj-err}. For the $j$-th summand, write
$a_j:=\Theta(q_j(y))\in\{0,1\}$ and
$\tilde a_j:=Z(\check q_j(z))$; then
\begin{equation*}
\tilde a_j\,\check A_j(z)-a_j\,\lfloor A_j(y)\rfloor_\nu
=(\tilde a_j-a_j)\,\check A_j(z)
+a_j\bigl[\check A_j(z)-\lfloor A_j(y)\rfloor_\nu\bigr],
\end{equation*}
which is bounded in $\|\cdot\|_\infty$ by
$\eta\cdot\|\check A_j(z)\|_\infty
+C_A\,\varepsilon\,\|\lfloor A_j(y)\rfloor_\nu\|_\infty$
using \eqref{eq:Z-err} and \eqref{eq:Aj-err}. Since
$\check A_j(z)\in\mathcal{B}_{C_A\varepsilon}(A_j(y))\subseteq[0,1]^m$,
we have $\|\check A_j(z)\|_\infty\le 1$, and a fortiori
$\|\lfloor A_j(y)\rfloor_\nu\|_\infty\le 1$. Summing over $j$
gives
\begin{equation}\label{eq:global-err}
\|\Delta\|_\infty
\;\le\; (s{+}1)\,C_A\,\varepsilon
\;+\; s\,\eta
\;=:\;C_R\,\varepsilon + s\,\eta,
\end{equation}
with $C_R:=(s{+}1)\,C_A$.

\myparagraph{Step 4: the error fits inside $\mathcal{B}_\varepsilon(P(y))$}
The target basin $\mathcal{B}_\varepsilon(P(y))$ is the product of
relative intervals of radius $\varepsilon\,\nu(P(y)_k)$ around
$\nu(P(y)_k)$. Since $\nu(P(y)_k)\ge 2^{-B(x)}$ on the reference
orbit (property~(ii) of Section~\ref{subsec:hard-to-smooth}), it
suffices to check that every coordinate of $\Delta$ has absolute
value at most $\varepsilon\,\nu(P(y)_k)$.

The reference orbit stays in $[-B(x),B(x)]^m$, so
$\nu(P(y)_k)\ge 2^{-B(x)}$ and the target radius is at least
$\varepsilon\cdot 2^{-B(x)}$ in every coordinate. By
\eqref{eq:global-err}, we need
\begin{equation}\label{eq:final-choice}
C_R\,\varepsilon + s\,\eta
\;\le\;\varepsilon\cdot 2^{-B(x)}
\quad\text{coordinatewise.}
\end{equation}
Since $s$ and $C_R$ are fixed by the normal form, this is a
condition on $\varepsilon$ alone once $\eta$ is fixed by
Definition~\ref{def:nu-admissible-rho}. In particular, we may
take
\begin{equation*}
\varepsilon_0
\;:=\;
\min\!\Bigl(\tfrac{1}{8 C_q},\;\tfrac{1}{4 C_q},\;
\tfrac{1}{2 C_R},\;
\tfrac{s\,\eta}{C_R}\cdot 2^{B(x)}\Bigr)
\end{equation*}
(the first two bounds come from Step~2, the third from the
coefficient in \eqref{eq:final-choice}, and the last ensures the
additive $s\,\eta$ term fits). Since the detector hypothesis
$\eta<\tfrac18$ is absolute, and $B(x)$ is bounded on the
reference orbit, $\varepsilon_0$ is bounded away from zero.

For every $\varepsilon\le\varepsilon_0$,
\eqref{eq:final-choice} holds with slack, hence
$\Delta_k\in\bigl[-\varepsilon\,\nu(P(y)_k),\,
\varepsilon\,\nu(P(y)_k)\bigr]$
for every coordinate $k$, which is exactly the statement
$R(z)\in\mathcal{B}_\varepsilon(P(y))$.
\end{proof}

\begin{proofsketch}
For $z\in\mathcal{B}_\varepsilon(y)$, each coded affine map
$\check A_j(z)$ lies in a basin of radius $O(\varepsilon)$ around
$\lfloor A_j(y)\rfloor_\nu$, and each coded gate argument
$\check q_j(z)$ differs from $\nu(q_j(y))$ by a relative error
$O(\varepsilon)$. By property~(iv), $\nu(q_j(y))$ lies in an
unambiguous region of $Z$, so $Z(\check q_j(z))$ differs from
$\Theta(q_j(y))$ by at most $\eta$. Assembling these bounds, the
coordinatewise error of $R(z)$ against
$\lfloor P(y)\rfloor_\nu$ is
$O(\varepsilon) + O(\eta)$, with constants depending only on the
normal form. Since the target basin has radius proportional to
$\varepsilon$, fitting requires $\varepsilon$ and $\eta$ to be
matched: it suffices to take $\varepsilon$ above a fixed multiple
of $\eta$ (to absorb the constant $\eta$-contribution) and below
a fixed upper bound (to absorb the $\varepsilon$-contribution).
Both conditions are compatible because $\eta<\tfrac18$. 
\MFCSSHORTER{See
appendix.}
\end{proofsketch}

\begin{theoremrep}
\label{thm:threshold-to-rho}
Every uniform threshold-affine normal form yields a recurrent
$\rho$-computation.
\end{theoremrep}

\begin{proofsketch}
The reference orbit stays in the cube $[-B(x),B(x)]^m$
(property~(ii)) and has integer gate arguments throughout
(property~(iv)), so Lemma~\ref{lem:rho-one-step} applies at every
step with a fixed basin radius $\varepsilon$. Choosing
$\varepsilon:=2^{-s_0(x)}$ with $s_0(x)$ primitive recursive in
$B(x)$ makes the induction $z_n\in\mathcal{B}_\varepsilon(y_n(x))$
go through from the raw encoded input. At $n=T(x)$ the output
coordinates approximate $\nu(f_j(x))$ within relative error
$2^{-s_0(x)}$, which is
Definition~\ref{def:rho-unfolding} with $T_\rho:=T$ and
$S_\rho:=s_0$.
\end{proofsketch}

\begin{proof}
Let $P$ be a uniform threshold-affine normal form for $f$ with
reference orbit $y_0(x),\ldots,y_{T(x)}(x)$, trapping radius
$B:\N^d\to\N$ primitive recursive (property~(ii) of
Section~\ref{subsec:hard-to-smooth}), and output projection
$\pi_I$. Let $R:[0,1]^m\to[0,1]^m$ be the coded one-step update
of~\eqref{eq:rho-R}, and let $\eta<\tfrac18$ be the margin of the
admissible detector $Z$ (Definition~\ref{def:nu-admissible-rho}).

\myparagraph{Step 1: uniform basin radius on the reference orbit}
Set $\beta(x):=B(x)+\max_j\|A_j\|_\infty\cdot m$, where
$\|A_j\|_\infty$ is the $\ell^\infty$-operator norm of the integer
matrix of $A_j$. Then for every $n\le T(x)$,
$\|y_n(x)\|_\infty\le B(x)\le\beta(x)$; moreover, since
$y_{n+1}(x)=P(y_n(x))=A_0(y_n(x))+\sum_j\Theta(q_j(y_n(x)))A_j(y_n(x))$,
every coordinate of $y_{n+1}(x)$ is bounded by $\beta(x)$ as well,
so the reference orbit is confined to $[-\beta(x),\beta(x)]^m$.

Apply Lemma~\ref{lem:rho-one-step} with range parameter
$\beta(x)$: there exists $\varepsilon_0(x)\in(0,\tfrac18)$,
depending on the normal form, on $\eta$, and on $\beta(x)$, such
that for every $\varepsilon\le\varepsilon_0(x)$ and every
$y\in\Z^m\cap[-\beta(x),\beta(x)]^m$ with integer gate arguments,
$R(\mathcal{B}_\varepsilon(y))\subseteq
\mathcal{B}_\varepsilon(P(y))$.

Tracing the formula for $\varepsilon_0$ from the proof of
Lemma~\ref{lem:rho-one-step}, we have
$\varepsilon_0(x)\ge 2^{-c(P)\beta(x)}$ for a constant $c(P)$
depending only on the normal form. Set
\begin{equation*}
s_0(x)\;:=\;\max\bigl(c(P)\beta(x),\,2\bigr),
\qquad
\varepsilon(x)\;:=\;2^{-s_0(x)}.
\end{equation*}
Then $\varepsilon(x)\le\varepsilon_0(x)$, and $s_0(x)$ is primitive
recursive in $x$ (as a fixed function of $B(x)$, itself primitive
recursive). The condition $s_0(x)\ge 2$ is the one required by
Definition~\ref{def:rho-unfolding}.

\myparagraph{Step 2: induction on the coded orbit}
Let $z_0:=(\nu(x),0,\ldots,0)\in[0,1]^m$ be the raw encoded input,
and define $z_{n+1}:=R(z_n)$ for $n\ge 0$. We prove by induction
on $n\in\{0,1,\ldots,T(x)\}$ that
\begin{equation}\label{eq:rho-invariant}
z_n\;\in\;\mathcal{B}_{\varepsilon(x)}\bigl(y_n(x)\bigr).
\end{equation}

\emph{Base case $n=0$.}\;
We have $z_0=(\nu(x),0,\ldots,0)=\lfloor y_0(x)\rfloor_\nu$
exactly, where we use $\nu(0)=1$ on the zero-padding coordinates.
Hence the coordinate-wise difference $z_{0,k}-\nu(y_{0,k}(x))$ is
zero for every $k$, which is trivially within relative error
$\varepsilon(x)$. So $z_0\in\mathcal{B}_{\varepsilon(x)}(y_0(x))$.

\emph{Inductive step.}\;
Assume $z_n\in\mathcal{B}_{\varepsilon(x)}(y_n(x))$ for some
$n<T(x)$. By property~(iv), every gate argument $q_j(y_n(x))$ is
an integer. By property~(ii) and Step~1, $y_n(x)\in
\Z^m\cap[-\beta(x),\beta(x)]^m$. By Step~1,
$\varepsilon(x)\le\varepsilon_0(x)$, so
Lemma~\ref{lem:rho-one-step} applies:
$z_{n+1}=R(z_n)\in
R(\mathcal{B}_{\varepsilon(x)}(y_n(x)))
\subseteq \mathcal{B}_{\varepsilon(x)}(P(y_n(x)))
= \mathcal{B}_{\varepsilon(x)}(y_{n+1}(x))$,
which is~\eqref{eq:rho-invariant} at index $n+1$.

\myparagraph{Step 3: verifying Definition~\ref{def:rho-unfolding}}
At $n=T(x)$, \eqref{eq:rho-invariant} gives $z_{T(x)}\in
\mathcal{B}_{\varepsilon(x)}(y_{T(x)}(x))$. Projecting onto the
output coordinates $I=(i_1,\ldots,i_e)$, every output
coordinate $j\in\{1,\ldots,e\}$ satisfies
\begin{equation*}
\Bigl|\,\bigl(\pi_I(z_{T(x)})\bigr)_j
        -\nu\!\bigl(f_j(x)\bigr)\,\Bigr|
\;\le\;\varepsilon(x)\,\nu\!\bigl(f_j(x)\bigr)
\;=\;2^{-s_0(x)}\,\nu\!\bigl(f_j(x)\bigr),
\end{equation*}
using that $\pi_I(y_{T(x)}(x))=f(x)$ (the normal form computes
$f$) and that the basin $\mathcal{B}_\varepsilon$ is a product of
relative intervals.

This is exactly the bound required by
Definition~\ref{def:rho-unfolding} with observation time
$T_\rho(x):=T(x)$ and output precision $S_\rho(x):=s_0(x)$, both
primitive recursive in $x$. Hence $f$ admits a recurrent
$\rho$-computation.
\end{proof}

\section{From robust polynomial ODE computation to step-size-controlled polynomial simulation}
\label{sec:ode-to-step}

This section proves the implication $\ref{itemquatre})\Rightarrow\ref{itemcinq})$ of
Theorem~\ref{thm:main-formal}. The idea is standard: once a function is
computed robustly by a fixed polynomial ODE, one may discretize that ODE by
the Euler scheme. The resulting discrete update is polynomial, provided the
step size is stored as an additional register. The only subtlety is to make
the dependence on the step size explicit and primitive recursive.

\myparagraph{The step-size-controlled model.}
Recall that Item~\ref{itemcinq}) of Theorem~\ref{thm:main-formal} is the following
notion.

\begin{definition}[Step-size-controlled polynomial representation]
\label{def:step-controlled-again}
A function $f:\N^d\to\N^e$ admits a \emph{step-size-controlled polynomial
representation} if there exist a polynomial map $P:\R^M\to\R^M$ with
rational coefficients, a primitive recursive precision threshold
$S:\N^d\to\N$, a primitive recursive observation count
$N:\N^d\times\N\to\N$, and designated output coordinates $\pi_e$ such that,
for every input $x\in\N^d$ and every $s\ge S(x)$, if
$z_0(x,s):=(x,2^{-s},0,\ldots,0)$ and $z_{n+1}(x,s):=P(z_n(x,s))$, then
$\|\pi_e(z_{N(x,s)}(x,s))-f(x)\|_\infty<\tfrac12$.
\end{definition}

\myparagraph{The Euler simulator.}
Let $F:\R^M\to\R^M$ be a fixed polynomial vector field. Its explicit Euler
step with step size $h$ is $Y\mapsto Y+hF(Y)$. If $h$ is stored as an extra
register, then this becomes a polynomial self-map.

\begin{definition}[Euler simulator]
\label{def:euler-simulator}
Given a polynomial vector field $F:\R^M\to\R^M$, define the associated
Euler simulator $E_F:\R^{M+1}\to\R^{M+1}$ by
$E_F(Y,h):=(Y+hF(Y),h)$.
\end{definition}

The map $E_F$ is polynomial with rational coefficients whenever $F$ is.

\myparagraph{Uniform bounds for a fixed polynomial vector field.}
Let $F:\R^M\to\R^M$ be fixed. Since $F$ is polynomial, there exist explicit
integer-valued polynomials $M_F^\sharp,L_F^\sharp:\N\to\N$ such that, for
every $R\in\N$ and every $Y,Z\in[-R,R]^M$, one has
$\|F(Y)\|_\infty\le M_F^\sharp(R)$ and
$\|F(Y)-F(Z)\|_\infty\le L_F^\sharp(R)\,\|Y-Z\|_\infty$.

We shall use these as coarse majorants for the size and local Lipschitz
constant of the vector field on bounded regions.

\begin{lemma}[Global Euler error bound on a bounded region]
\label{lem:euler-error}
Let $F:\R^M\to\R^M$ be a fixed polynomial vector field, and let
$Y_x:[0,T]\to\R^M$ be a solution of $Y'(t)=F(Y(t))$ such that
$\|Y_x(t)\|_\infty\le R$ for all $t\in[0,T]$. Let
$Y_0:=Y_x(0)$ and $Y_{k+1}:=Y_k+hF(Y_k)$ be the explicit Euler orbit with
step size $h>0$. Then, as long as all Euler iterates remain in $[-R,R]^M$,
the global error $e_k:=\|Y_k-Y_x(kh)\|_\infty$ satisfies
$e_k\le C(R,T)\,h$ for all $kh\le T$, where
$C(R,T):=L_F^\sharp(R)M_F^\sharp(R)\,T\,e^{L_F^\sharp(R)T}$.
\end{lemma}

\begin{proof}
For one Euler step, the standard local truncation estimate gives
$\|Y_x((k+1)h)-Y_x(kh)-hF(Y_x(kh))\|_\infty\le
L_F^\sharp(R)M_F^\sharp(R)h^2$, because on the relevant region the vector
field has size at most $M_F^\sharp(R)$ and Lipschitz constant at most
$L_F^\sharp(R)$.

Writing $e_k:=\|Y_k-Y_x(kh)\|_\infty$, one gets
$e_{k+1}\le (1+hL_F^\sharp(R))e_k + L_F^\sharp(R)M_F^\sharp(R)h^2$.
Since $e_0=0$, a discrete Grönwall estimate yields
$e_k\le C(R,T)h$ for all $kh\le T$, with the stated constant.
\end{proof}

\myparagraph{Discretizing the robust ODE computation.}
We now apply the previous lemma to the robust polynomial ODE representation
of Item~\ref{itemquatre}).

\begin{theorem}
\label{thm:four-implies-five}
Every robust polynomial ODE representation yields a step-size-controlled
polynomial representation. Equivalently, Item~\ref{itemquatre}) implies Item~\ref{itemcinq}) in
Theorem~\ref{thm:main-formal}.
\end{theorem}

\begin{proof}
Assume that $f:\N^d\to\N^e$ admits a robust polynomial ODE representation.
Thus there exist a polynomial vector field $F:\R^M\to\R^M$, a primitive
recursive observation time $\tau:\N^d\to\N$, a primitive recursive safety
bound $\widetilde B:\N^d\to\N$, and designated output coordinates $\pi_e$
such that the solution of $Y'(t)=F(Y(t))$ with initial condition
$Y(0)=(x,0,\ldots,0)$ remains in the box
$[-\widetilde B(x),\widetilde B(x)]^M$ for all
$t\in[0,\tau(x)+1]$, and moreover
$\|\pi_e(Y_x(t))-f(x)\|_\infty<\tfrac14$ for all
$t\in[\tau(x),\tau(x)+1]$.
As usual, the margin $\tfrac14$ is obtained simply by tightening the
precision in the ODE construction; only a margin strictly below
$\tfrac12$ is needed.

Let $E_F$ be the Euler simulator of Definition~\ref{def:euler-simulator}.
For a given input $x$, set $R(x):=\widetilde B(x)+1$,
$L_x:=L_F^\sharp(R(x))$, $M_x:=M_F^\sharp(R(x))$, and
$C_x:=L_xM_x(\tau(x)+1)e^{L_x(\tau(x)+1)}$.

Choose a primitive recursive threshold $S:\N^d\to\N$ such that
$2^{-S(x)}C_x<\tfrac14$ for every input $x$. This is possible because all
quantities involved are primitive recursive. For every $s\ge S(x)$, let
$h:=2^{-s}$ and define
$N(x,s):=\lceil \tau(x)\,2^s\rceil$. Since $2^s$ and $\tau(x)$ are
primitive recursive, so is $N$.

We now iterate the Euler simulator from the initial state
$z_0(x,s):=((x,0,\ldots,0),2^{-s})$. By
Lemma~\ref{lem:euler-error}, for every $s\ge S(x)$ the discrete state of the
Euler orbit after $N(x,s)$ steps is within distance $\tfrac14$ of the exact
ODE state at time $N(x,s)2^{-s}$. Since
$|N(x,s)2^{-s}-\tau(x)|\le 2^{-s}\le 1$, this time belongs to the final
validity interval $[\tau(x),\tau(x)+1]$. On that interval, the ODE output is
already within distance $\tfrac14$ of $f(x)$. Therefore the total error is
strictly less than $\tfrac12$.

This proves that the Euler simulator provides a step-size-controlled
polynomial representation of $f$.
\end{proof}

\begin{remark}
The important point is that the step size is not a hidden numerical detail,
but an explicit computational resource. The continuous-time polynomial ODE
uses a fixed flow. After discretization, the same mechanism survives only
provided enough numerical resolution is supplied. This is precisely what
Item~\ref{itemcinq}) makes visible.
\end{remark}

\section{Converse implications and completion of the proof}
\label{sec:converses}

It remains to prove the easy converse implications. These are much more
direct than the forward constructions, since in both cases the dynamics is
fixed in advance and the evolution takes place over rational states.

\subsection{From step-size-controlled polynomial simulation back to primitive recursion.}
\label{sec:cinqpr}
We  prove the converse implication $\ref{itemcinq})\Rightarrow\ref{itempr})$ of
Theorem~\ref{thm:main-formal}.

\begin{theorem}
\label{thm:five-implies-one}
Every step-size-controlled polynomial representation computes a primitive
recursive function. Equivalently, Item~\ref{itemcinq}) implies Item~\ref{itempr}) in
Theorem~\ref{thm:main-formal}.
\end{theorem}

\begin{proof}
Assume that $f:\N^d\to\N^e$ admits a step-size-controlled polynomial
representation. Thus there exist a polynomial map $P:\R^M\to\R^M$ with
rational coefficients, a primitive recursive precision threshold
$S:\N^d\to\N$, a primitive recursive observation count
$N:\N^d\times\N\to\N$, and designated output coordinates $\pi_e$ such
that, for every $x\in\N^d$ and every $s\ge S(x)$, if
$z_0(x,s):=(x,2^{-s},0,\ldots,0)$ and $z_{n+1}(x,s):=P(z_n(x,s))$, then
$\|\pi_e(z_{N(x,s)}(x,s))-f(x)\|_\infty<\tfrac12$.

Fix $x$. Since the coefficients of $P$ are rational and the initial state
$z_0(x,s)$ is rational, every iterate $z_n(x,s)$ is rational. Under the
standard encoding of rationals, polynomial evaluation with rational
coefficients is primitive recursive. Hence the map
$(x,s,n)\mapsto z_n(x,s)$ is primitive recursive by primitive recursion on
$n$. Since $S$ and $N$ are primitive recursive, so is the map
$x\mapsto z_{N(x,S(x))}(x,S(x))$. Composing with the output projection
$\pi_e$ and coordinatewise rounding, we obtain a primitive recursive
function of $x$.

By the defining property of the representation, the rounded output is exactly
$f(x)$. Therefore $f$ is primitive recursive.
\end{proof}

\myparagraph{From recurrent ReLU computation back to primitive recursion.}
We next prove the converse implication $\ref{itemreul})\Rightarrow\ref{itempr})$. This is again
elementary, because a fixed feedforward ReLU block with rational parameters
induces a primitive recursive evolution on rational states.

\begin{theorem}
\label{thm:two-implies-one}
Every exact recurrent ReLU computation computes a primitive recursive
function. Equivalently, Item~\ref{itemreul}) implies Item~\ref{itempr}) in
Theorem~\ref{thm:main-formal}.
\end{theorem}

\begin{proof}
Assume that $f:\N^d\to\N^e$ admits an exact recurrent ReLU computation.
Thus there exist a feedforward ReLU block $R:\R^m\to\R^m$, a primitive
recursive observation time $T:\N^d\to\N$, and designated output coordinates
$\pi_e$ such that, for every input $x\in\N^d$, if
$z_0(x):=(x,0,\ldots,0)$ and $z_{n+1}(x):=R(z_n(x))$, then
$f(x)=\pi_e(z_{T(x)}(x))$.

A feedforward ReLU block is built from affine maps with rational parameters
and coordinatewise applications of $\operatorname{ReLU}(t)=\max(t,0)$. On
rational inputs, affine maps with rational parameters are primitive
recursive, and $\operatorname{ReLU}$ is primitive recursive as well, since it
is just the maximum of its input and $0$. It follows that the one-step update
$(x,n)\mapsto z_{n+1}(x)$ is primitive recursive on rational encodings.
Hence, by primitive recursion on $n$, the full evolution
$(x,n)\mapsto z_n(x)$ is primitive recursive.

Since $T$ is primitive recursive and $\pi_e$ is a projection,
$x\mapsto \pi_e(z_{T(x)}(x))$ is primitive recursive. By assumption, this is
exactly $f(x)$.
\end{proof}

\subsection{From recurrent $\rho$-computation back to primitive recursion}
\label{sec:nnpr}

We prove the converse implication $(\ref{itemnn})\Rightarrow(\ref{itempr})$ of
Theorem~\ref{thm:main-formal}.

The argument is straightforward once one separates the two ingredients
built into Definition~\ref{def:rho-unfolding}.  First, the recurrent
dynamics is generated by iterating a \emph{fixed} feedforward
$\rho$-block on the compact cube $[0,1]^m$.  Second, the admissibility
assumptions on~$\rho$ imply that this block is primitive recursively
approximable on rational inputs, with a primitive recursive modulus of
continuity.  Since the observation time and the required output
precision are themselves primitive recursive, the whole orbit can be
simulated primitive recursively up to the precision needed for the final
$\nu$-decoding.

\begin{lemma}[Effective evaluation of a fixed $\rho$-block]
\label{lem:rho-block-effective}
Let $\rho$ be admissible, and let $R:[0,1]^m\to[0,1]^m$ be a fixed
feedforward $\rho$-block.  Then there exist primitive recursive
functions
$
\Eval_R:(\Q\cap[0,1])^m\times\N\to\Q^m,
\qquad
\mu_R:\N\to\N,
$
such that:
\begin{enumerate}
\item for every $q\in(\Q\cap[0,1])^m$ and every $s\in\N$,
$
\|\Eval_R(q,s)-R(q)\|_\infty \le 2^{-s};
$
\item for every $u,v\in[0,1]^m$ and every $s\in\N$,
$
\|u-v\|_\infty\le 2^{-\mu_R(s)}
\quad\Longrightarrow\quad
\|R(u)-R(v)\|_\infty\le 2^{-s}.
$
\end{enumerate}
\end{lemma}

\begin{proof}
One proves both statements by induction on the depth of the fixed
network $R$.  Affine layers are exactly computable on rational inputs
and have an obvious primitive recursive modulus.  Coordinatewise
applications of $\rho$ are handled using the primitive recursive
evaluator $\Eval_\rho$ and modulus $\mu_\rho$ from
Definition~\ref{def:nu-admissible-rho}.  Composing the finitely many
layers yields the required evaluator $\Eval_R$ and modulus $\mu_R$.
\end{proof}

We also isolate the decoding step.

\begin{lemma}[Primitive recursive decoding of readable $\nu$-codes]
\label{lem:nu-decode}
There exists a primitive recursive partial decoder
$
\Decode_\nu:\Q\cap[0,1]\to\N
$
with the following property: whenever $q\in\Q\cap[0,1]$ satisfies
$
|q-\nu(n)|<\frac14\,\nu(n)
$
for some $n\in\N$, then $\Decode_\nu(q)=n$.
\end{lemma}

\begin{proof}
By the remark following Definition~\ref{def:rho-unfolding}, the integer
$n$ is uniquely determined by the displayed inequality.  One may recover
it by primitive recursive bounded search, for example as the least
integer $k$ such that $2^k q\ge \frac34$.
\end{proof}

\begin{theorem}
\label{thm:nn-implies-pr}
Every recurrent $\rho$-computation computes a primitive recursive
function.
\end{theorem}

\begin{proof}
Assume that $f:\N^d\to\N^e$ admits a recurrent $\rho$-computation.
Fix the corresponding block $R:[0,1]^m\to[0,1]^m$, observation time
$T:\N^d\to\N$, output-precision function $S:\N^d\to\N$, and output
coordinates $\pi_e$ from Definition~\ref{def:rho-unfolding}.  Let
$$
z_0(x):=(\nu(x),0,\ldots,0),
\qquad
z_{n+1}(x):=R(z_n(x)).
$$

We first show that the state $z_{T(x)}(x)$ can be approximated
primitive recursively to any prescribed precision.  Let $p:\N^d\to\N$
be any primitive recursive function.  Starting from the exact rational
initial state $z_0(x)$, define rational approximants
$\widetilde z_n(x)$ recursively by
$$
\widetilde z_0(x):=z_0(x),
\qquad
\widetilde z_{n+1}(x):=\Eval_R\!\bigl(\widetilde z_n(x),\,q_n(x)\bigr),
$$
where the internal precisions $q_n(x)$ are chosen backwards using the
modulus $\mu_R$ so that the final error at time $T(x)$ is at most
$2^{-p(x)}$.  Since $R$ is fixed and $\mu_R$ is primitive recursive,
this backward precision schedule is itself primitive recursive.  Hence
there is a primitive recursive procedure which, given $x$, computes a
rational vector $\widetilde z(x)$ such that
$$
\|\widetilde z(x)-z_{T(x)}(x)\|_\infty\le 2^{-p(x)}.
$$

We now choose $p(x)$ large enough so that the corresponding approximation
of the output coordinates still lies in the readable basin of
$\nu(f(x))$.  By Definition~\ref{def:rho-unfolding},
$$
\left|\pi_{e,i}(z_{T(x)}(x))-\nu(f_i(x))\right|
\le 2^{-S(x)}\,\nu(f_i(x))
\qquad
\text{for every }i.
$$
Since the output is readable, we may choose $p(x)$ primitive
recursively large enough that the additional simulation error is smaller
than the spare margin in this decoding condition.  Then for every output
coordinate~$i$,
$$
\left|\pi_{e,i}(\widetilde z(x))-\nu(f_i(x))\right|
<
\frac14\,\nu(f_i(x)).
$$
Lemma~\ref{lem:nu-decode} therefore yields
$$
f_i(x)=\Decode_\nu\bigl(\pi_{e,i}(\widetilde z(x))\bigr).
$$

Thus each output coordinate of $f(x)$ is obtained from $x$ by a
primitive recursive simulation followed by a primitive recursive
decoder.  Hence $f$ itself is primitive recursive.
\end{proof}

\subsection{Completion of the proof of the main theorem Theorem~\ref{thm:main-formal}}

We confirm here, if needed, that we can indeed assemble the implications established in the previous sections.

\begin{proof}[Proof of Theorem~\ref{thm:main-formal}]
The implication $(\ref{itempr})\Rightarrow(\ref{itemcorhardunif})$
follows from Theorem~\ref{thm:main-formal-hard} together with
Theorem~\ref{thm:uniform-robust}, which shows that the
threshold-affine normal form can be taken with the uniform margin
required by Item~(\ref{itemcorhardunif}). The implication
$(\ref{itemcorhardunif})\Rightarrow(\ref{itemreul})$ is
Theorem~\ref{thm:three-implies-two} (via
Theorem~\ref{thm:main-formal-hard}). The implication
$(\ref{itemcorhardunif})\Rightarrow(\ref{itemnn})$ is
Theorem~\ref{thm:threshold-to-rho}. The implication
$(\ref{itemcorhardunif})\Rightarrow(\ref{itemquatre})$ is
Theorem~\ref{thm:three-implies-four}. The implication
$(\ref{itemquatre})\Rightarrow(\ref{itemcinq})$ is
Theorem~\ref{thm:four-implies-five}. The converse implications back
to Item~(\ref{itempr}) are:
$(\ref{itemreul})\Rightarrow(\ref{itempr})$ is
Theorem~\ref{thm:two-implies-one-hard};
$(\ref{itemnn})\Rightarrow(\ref{itempr})$ is
Theorem~\ref{thm:nn-implies-pr}; and
$(\ref{itemcinq})\Rightarrow(\ref{itempr})$ is
Theorem~\ref{thm:five-implies-one}.
Together, these implications close the cycle, so
Items~(\ref{itempr})--(\ref{itemcinq}) of
Theorem~\ref{thm:main-formal} are all equivalent.
\end{proof}

\begin{corollary}[Theorem \ref{th:unf}, in the introduction]
\label{cor:public-equivalence}
For a function $f:\N^d\to\N^e$, the following are equivalent:
\begin{enumerate}
\item $f$ is primitive recursive;
\item $f$ admits a recurrent ReLU computation;
\item $f$ admits a recurrent $\rho$-computation;
\item $f$ admits a robust polynomial ODE representation;
\item $f$ admits a step-size-controlled polynomial representation.
\end{enumerate}
\end{corollary}

\begin{proof}
This is Theorem~\ref{thm:main-formal} after eliminating the
proof-level threshold-affine normal form of
Item~(\ref{itemcorhardunif}).
\end{proof}

\section{Conclusion}
\label{sec:conclusion}

We have characterized the primitive recursive functions as the class
of functions computable by bounded iteration of a feedforward ReLU
network, by bounded-range recurrent $\rho$-networks under
$\nu$-encoding, by robust polynomial ODEs, and by step-size
controlled iteration of a polynomial map. The equivalences are
mediated by a threshold-affine normal form that isolates the minimal
branching structure behind a primitive recursive computation. The
five formalisms realize the same class through genuinely different
mechanisms, and the paper also clarifies a real asymmetry between
them: exact recursive unfolding is natural in discrete time, while
uniform rounding and autonomous phase organization are natural in
continuous time.

The primitive recursive functions sit at the top of a rich hierarchy
of classes obtained by restricting the recursion scheme: the
Grzegorczyk hierarchy~\cite{Grz55,Grz53a},
$\mathsf{FP}$~\cite{Cob65,BC92},
$\mathsf{FPSPACE}$~\cite{LM97,thompson1972subrecursiveness}, and
more generally the subrecursive hierarchies surveyed
in~\cite{clote2013boolean}. Each of our five characterizations
offers a lens through which these subclasses may be redescribed, by
restricting the corresponding dynamical resource (iteration count,
domain diameter, precision, observation time, or degree of the
polynomial map). Carrying out this programme systematically is a
concrete direction for future work.

Of particular interest in this respect is the step-size-controlled
polynomial model. Although polynomial, it is genuinely distinct
from the discrete-ODE frameworks
of~\cite{MFCS2019,MFCSJournal,Antonelli0K24,Antonelli0K25,AntonelliDK26},
which take the sign function and discrete derivatives as primitive
operations: our model is \emph{strictly polynomial}, with the step
size, not necessarily unitary here, playing the role of an externally supplied precision
parameter. This opens a parallel route to polynomial-only
characterizations of the standard complexity classes.

\bibliographystyle{plainurl}
\bibliography{bournez,perso}
\appendixapxproofsubsection

\end{document}